\documentclass[11pt,a4paper]{article}
\usepackage{jheppub}
\makeatletter
\gdef\@fpheader{}
\makeatother
\usepackage[T1]{fontenc}
\usepackage[utf8]{inputenc}
\usepackage{lmodern}
\usepackage{microtype}
\usepackage{amsthm,mathtools,booktabs,array,longtable}
\usepackage{xcolor,enumitem}
\hypersetup{pdfauthor={Lin Mai and Yaobo Zhang},bookmarksnumbered=true}
\newcommand{\C}{\mathbb C}

\newcommand{\ang}[2]{\langle #1#2\rangle}
\newcommand{\sqr}[2]{[#1#2]}

\newcommand{\Sym}{\operatorname{Sym}}

\newcommand{\sgn}{\operatorname{sgn}}

\newcommand{\rank}{\operatorname{rank}}

\setlist[enumerate]{itemsep=2pt,topsep=5pt}
\setlist[itemize]{itemsep=2pt,topsep=5pt}
\renewcommand{\arraystretch}{1.18}
\numberwithin{equation}{section}
\allowdisplaybreaks[1]
\newtheorem{theorem}{Theorem}[section]
\newtheorem{lemma}[theorem]{Lemma}
\newtheorem{proposition}[theorem]{Proposition}
\newtheorem{corollary}[theorem]{Corollary}
\theoremstyle{definition}
\theoremstyle{remark}
\hypersetup{pdftitle={Algebraic versus physical uniqueness of MHV gravity numerators}}
\title{Algebraic versus physical uniqueness\\of MHV gravity numerators}
\author{Lin Mai}
\author{and Yaobo Zhang}
\affiliation{School of Physics, Ningxia University,\\Yinchuan 750021, China}
\emailAdd{120251307670@stu.nxu.edu.cn}
\emailAdd{yaobozhang@nxu.edu.cn}
\abstract{We give explicit determinant counterexamples to the uniqueness of MHV gravity numerators from their degrees and simultaneous pair zeros. At seven points the solution space has dimension seven; Bose symmetry removes six directions and leaves the Hodges line. At eight points a two-dimensional Bose-symmetric space remains. Regular integer-spinor configurations distinguish these deformations from the Hodges numerator. On the eight-point plane, collinear factorization, the leading positive-helicity soft coefficient and ordinary BCFW decay each select the Hodges line. For $n\ge5$, we prove two results on the full fixed-degree, common-denominator ansatz, without pair-zero conditions. Complete collinear residue functions on all pairs determine at most one numerator; Bose symmetry makes one channel sufficient. The splitting law and a normalized lower-point amplitude supply the boundary function, starting from an explicit five-point seed. With Bose symmetry, $O(Z^{-1})$ decay under one fixed same-helicity shift at generic conserving kinematics also selects the Hodges line. Normalized collinear and soft data fix the eight-point scale; the homogeneous shift condition leaves one constant per multiplicity.}
\keywords{Scattering Amplitudes, Classical Theories of Gravity,
Differential and Algebraic Geometry, Discrete Symmetries}
\arxivnumber{2608.11792}
\begin{document}
\maketitle
\flushbottom
\section{The numerator space and its constraints}
\label{sec:introduction}\label{sec:setup}
Koefler et al.~\cite{KoeflerEtAl2024} proposed that the MHV degrees and simultaneous pair zeros determine the tree-level gravity numerator up to scale. Their five- and six-point calculations give a line. We classify the full seven-point space and the Bose-symmetric eight-point sector by explicit determinant bases. Bose symmetry removes the extra seven-point directions, but an additional direction survives at eight points.

Collinear residues, leading soft coefficients and BCFW decay each select the Hodges line on the eight-point plane. For $n\ge5$, complete collinear data determine at most one numerator in the full prescribed-degree common-denominator ansatz: one channel suffices with Bose symmetry, and all channels suffice without it. The splitting law and a normalized lower-point amplitude supply this boundary function, starting from the five-point seed in Section~\ref{subsec:toy}. One fixed same-helicity shift with $O(Z^{-1})$ decay on generic complex conserving kinematics likewise selects the Hodges line in the full Bose-symmetric ansatz. Both all-multiplicity results hold without pair-zero conditions.

\subsection{Kinematics and the candidate spaces}
Take all momenta outgoing and complex, $p_i=X_iY_i^{\mathsf T}$, with $X_i=(x_i,y_i)^{\mathsf T}$ and $Y_i=(u_i,v_i)^{\mathsf T}$. The angle and square brackets are the determinants of two $X$ and two $Y$ columns, respectively. Strip the coupling, the momentum delta function and $\ang{r}{s}^8$ for the two negative-helicity gravitons. With the conventional common denominator of \cite{KoeflerEtAl2024}, write
\begin{equation}
  \mathcal A_n=N_n/D_n,\qquad D_n=\prod_{i<j}\ang{i}{j}.
  \label{eq:common-denominator}
\end{equation}
A candidate $N_n$ is a Lorentz-invariant polynomial modulo the four momentum equations:
\begin{equation}
 \begin{gathered}
 S_n=\C[x_i,y_i,u_i,v_i\mid 1\le i\le n],\qquad
 \mathfrak M_n=\Big(\sum_iX_iY_i^{\mathsf T}\Big),\\
 Q_n=(S_n/\mathfrak M_n)^G,\qquad G=SL_2(\C)_X\times SL_2(\C)_Y .
 \end{gathered}
 \label{eq:Q-definition}
\end{equation}
Each leg of $\mathcal A_n$ has weight $-4$ under $(X_i,Y_i)\mapsto(t_iX_i,t_i^{-1}Y_i)$. The stripped mass dimension is $2-8=-6$. These conditions fix the numerator's angle and square degrees $(A_n,B_n)$, its leg weights $\chi_i$, and its spinor multigrading $\tau_n$:
\begin{equation}
 A_n=\binom n2-n-3,\quad B_n=n-3,\quad \chi_i=n-5,\qquad
 \tau_n=(2A_n,2B_n;(n-5)^n).
 \label{eq:target-degree}
\end{equation}
Write $V_n=(Q_n)_{\tau_n}$ for this space. The label-permutation group acts on the denominator by its sign, so $V_n^{\sgn}$ consists of the numerators of Bose-symmetric candidates. We use $n\ge5$ throughout; see \cite{SpinorVarieties} for the spinor-helicity variety in these coordinates.

\paragraph{A numerator zero and a residue.}
For a fixed pair put $x=\ang{a}{b}$, $y=\sqr{a}{b}$, $D_n=x\widehat D_{ab}$. The proposed zero means membership in the pair ideal: $N_n=xU+yV$. At a generic point on the branch $x=0$ one obtains
\begin{equation}
  \left.x\mathcal A_n\right|_{x=0}
    =\left.\frac{yV}{\widehat D_{ab}}\right|_{x=0}.
  \label{eq:local-residue}
\end{equation}
The pair condition fixes the factor $y$ but leaves the coefficient $V/\widehat D_{ab}$ undetermined. Since the denominator also vanishes at $x=y=0$, this condition differs from an amplitude zero at regular kinematics. We write
\begin{equation}
 P_{ij}=(\ang{i}{j},\sqr{i}{j})S_n,\quad I_n=\bigcap_{i<j}P_{ij},\quad
 J_{ij}=(\ang{i}{j},\sqr{i}{j})Q_n,\quad W_n=V_n\cap\bigcap_{i<j}J_{ij}.
 \label{eq:target-spaces}
\end{equation}
Thus two polynomials differing by a momentum equation represent the same candidate. A determinant in $I_n$ already satisfies the pair conditions before the momentum quotient.

The Hodges numerator $N_n^{\mathrm H}=D_n\det'\Phi_n$ provides a reference solution \cite{Hodges2012}. The off-diagonal entries of $\Phi$ are $[ij]/\ang ij$; Appendix~\ref{app:hodges-tech} gives the reduced determinant and its polynomial numerator. Uniqueness means that $W_n$, or its Bose-symmetric sector $W_n^{\sgn}$, is the line spanned by $N_n^{\mathrm H}$.

\subsection{Zero conditions in amplitude reconstruction}
\label{subsec:bootstrap-context}
The reconstruction of Yang--Mills and gravity amplitudes from singularities and gauge invariance uses locality and minimal power counting~\cite{AHRT}. Here the common denominator, numerator degree and little-group weights are fixed, and we ask how much the simultaneous pair zeros determine. The Hodges determinant supplies a reference solution, following twistor-space and tree representations of MHV gravity amplitudes~\cite{MasonSkinner2010,NguyenEtAl2010,Hodges2012}.

Hidden zeros, enhanced large-shift behavior and smooth splitting give other uses of zeros in amplitude reconstruction~\cite{ArkaniHamedEtAl2024HiddenZeros,Rodina2025HiddenZeros,JonesParanjape2025}. The eight-point counterexample below preserves Bose symmetry while exhibiting the residue freedom in~\eqref{eq:local-residue}.

Algebraic-geometric ans\"atze encode vanishing and singularity conditions through ideals and their powers~\cite{DeLaurentisPage2022}. We impose ideal membership in the momentum quotient, with pole orders bounded by the fixed denominator. Related inverse constructions recover CHY integrands from factorization channels, or one-loop BCJ numerator families from generalized cuts~\cite{LiZhang2026InverseCHY,MaiZhang2026}. Their admissible functions and boundary data differ from ours: distinct classes in $W_n$ give distinct rational functions, whose physical validity requires further constraints.

\paragraph{The proposed lower-point induction.}
The proposed induction in~\cite{KoeflerEtAl2024} aligns two angle spinors, $X_{n-1}=cX_n$, and compares the boundary numerator with a fused lower-point numerator. The relevant degrees agree:
\begin{equation}
 (A_n-A_{n-1},B_n-B_{n-1})=(n-2,1)
 =\deg\left([n{-}1,n]\prod_{i=1}^{n-2}\ang in\right).
 \label{eq:context-degree-match}
\end{equation}
The Hodges numerator has this factor and the expected lower-point remainder, with the power of $c$ and sign in~\eqref{eq:fused-numerator}. For a general candidate the degree match still requires three proofs: divisibility by this factor, membership of the remainder in the lower-point space, and uniqueness of the extension away from the boundary. The extra seven-point directions obstruct an induction from the low-point dimensions alone. Our complete-boundary argument proves uniqueness of extension when tree factorization and a normalized lower-point amplitude supply the entire boundary function.

\subsection{The same candidate space in different algebraic coordinates}
\label{subsec:algebra-context}
Brackets and spinor components describe the same Lorentz-invariant functions, with Pl\"ucker identities and momentum conservation supplying the bracket relations. The spinor-helicity variety has associated flag-variety coordinates~\cite{SpinorVarieties}. Standard-monomial theory gives bases of fixed-degree quotient components, and Young tableaux organize the permutation representations~\cite{BrionLakshmibai2003,Fulton1997}. These tools underlie the symmetry-resolved seven-point rank calculations~\cite{MaiZhangV1}. Bose symmetry and lifting before the momentum quotient remain separate questions.

The completeness proofs use detection, interpolation and divisibility. At seven points the incidence data $k_P=\sum_{a\in P}c_a$ exhibit the sign-twisted permutation module: the constant coefficient vector gives the Hodges line, and the zero-sum vectors give six extra directions. At eight points the source and boundary relations yield the bound $4+2-4=2$. Exact computations check conventions and supply the separate nine-point certificate.

\paragraph{An affine picture of the pair ideal.}
On the dense chart
\begin{equation}
 X_i=a_i(1,z_i)^{\mathsf T},\qquad
 Y_i=b_i(1,w_i)^{\mathsf T},\qquad a_ib_i\ne0,
 \label{eq:context-affine-chart}
\end{equation}
we have $\ang ij=a_ia_j(z_j-z_i)$ and $[ij]=b_ib_j(w_j-w_i)$. After inverting the scale factors, a free pair ideal becomes $(z_j-z_i,w_j-w_i)$: the two labelled points of an affine plane coincide. Writing $c_i=a_ib_i$, momentum conservation becomes
\begin{equation}
 \sum_i c_i=\sum_i c_i z_i=\sum_i c_i w_i=\sum_i c_i z_iw_i=0.
 \label{eq:context-moments}
\end{equation}
Thus the ordinary diagonal arrangement is restricted to a four-equation moment subvariety. Moreover, the target bracket degree has the fixed deficit
\begin{equation}
 A_n+B_n=\binom n2-6.
 \label{eq:context-deficit}
\end{equation}
Work on generators of the radical big-diagonal ideal and their $q,t$ grading provides relevant combinatorial models~\cite{LeeLi2009Notes,LeeLi2011Catalan}. Counting $W_n$ also requires~\eqref{eq:context-moments}, fixed individual leg weights, and extension across the omitted zero-spinor strata. In particular, intersection with all pair ideals and quotienting by momentum need not commute. Our classification at seven points and in the alternating eight-point sector proves the existence of a simultaneous free representative in these degrees.

\subsection{Determinant zeros and linear constraints}
A row $r_i$ made of monomials of fixed degrees $(p,q)$ in $(X_i,Y_i)$ obeys $r_j=\lambda^p\mu^q r_i$ whenever $X_j=\lambda X_i$ and $Y_j=\mu Y_i$. A determinant containing the two rows then vanishes. Appending derivative rows gives the square matrices used in Section~\ref{sec:lowpoints} while retaining every original row.

\begin{lemma}[Testing a line in a two-dimensional space]\label{lem:plane-test}
Suppose a classified two-dimensional space has basis $H,R$, and $L_j$ is a boundary coefficient at configuration $j$, linear in the numerator. If $L_j(H)\ne0$, define $r_j=L_j(R)/L_j(H)$. For $N=\alpha H+\beta R$, agreement with the boundary of $H$ up to a common scale at two configurations requires
\begin{equation}
 \frac{L_1(N)}{L_1(H)}-\frac{L_2(N)}{L_2(H)}
   =\beta(r_1-r_2)=0.
 \label{eq:two-tests}
\end{equation}
Different $r_1,r_2$ force $\beta=0$. If the test prescribes a zero, one configuration with $L(H)=0$ and $L(R)\ne0$ suffices. A nonzero normalized value then determines $\alpha$. 
\end{lemma}
For an unrestricted numerator, uniqueness from complete collinear data follows from divisibility: a homogeneous polynomial vanishing on every angle divisor is divisible by their product $D_n$ (Corollary~\ref{cor:all-angle-divisibility}), whose degree exceeds the MHV angle degree.

\subsection{Five points: from a pair zero to a normalized residue}
\label{subsec:toy}
At five points, \eqref{eq:target-degree} gives bracket degrees $(A_5,B_5)=(2,2)$ and zero little-group weight on each leg. Consider the two-term ansatz
\begin{equation}
 r_i=(\rho_i,\sigma_i)
 =\bigl(\ang{1}{i}\sqr{2}{i},\ang{2}{i}\sqr{1}{i}\bigr),\qquad
 N=a\rho_5\sigma_4+b\rho_4\sigma_5.
 \label{eq:toy-ansatz}
\end{equation}
Here $i=4,5$. Each product has two angle and two square brackets, with every label appearing equally often in each kind; label $3$ is absent. Both products therefore have the required degrees and weights.

If $X_4=cX_5$ and $Y_4=dY_5$, then $\rho_4=cd\rho_5$ and $\sigma_4=cd\sigma_5$. The two products agree, so their difference vanishes:
\begin{equation}
 H=\rho_5\sigma_4-\rho_4\sigma_5
 =\det\begin{pmatrix}\rho_5&\sigma_5\\\rho_4&\sigma_4\end{pmatrix}.
 \label{eq:toy-determinant}
\end{equation}
Wherever the common product is nonzero on this locus, the pair condition forces $b=-a$.

\paragraph{The corresponding Hodges minor.}
To identify $H$ with the normalized Hodges numerator, delete rows $R=(1,2,3)$ and columns $C=(3,4,5)$ in \eqref{eq:reduced-determinant}, keeping both complements in increasing order. The cofactor sign is $(-1)^{6+12}=+1$, and the remaining matrix is
\begin{equation}
 \Phi^{123}_{345}
 =\begin{pmatrix}
 \sqr{1}{4}/\ang{1}{4}&\sqr{2}{4}/\ang{2}{4}\\
 \sqr{1}{5}/\ang{1}{5}&\sqr{2}{5}/\ang{2}{5}
 \end{pmatrix}.
 \label{eq:toy-hodges-minor}
\end{equation}
This minor uses only reference-independent off-diagonal entries. With $\Delta_{abc}=\ang{a}{b}\ang{b}{c}\ang{c}{a}$, cancellation of the six common angle factors gives
\begin{equation}
 \frac{D_5}{\Delta_{123}\Delta_{345}}
 =\ang{1}{4}\ang{1}{5}\ang{2}{4}\ang{2}{5}.
 \label{eq:toy-hodges-cancellation}
\end{equation}
The two minus signs from $\ang{3}{1}=-\ang{1}{3}$ and $\ang{5}{3}=-\ang{3}{5}$ cancel. Expanding the two-by-two determinant therefore yields
\begin{equation}
 \begin{aligned}
 N_5^{\mathrm H}=D_5\det'\Phi_5
 &=\ang{1}{5}\ang{2}{4}\sqr{1}{4}\sqr{2}{5}\\
 &\quad-\ang{1}{4}\ang{2}{5}\sqr{1}{5}\sqr{2}{4}=H.
 \end{aligned}
 \label{eq:toy-hodges-identity}
\end{equation}
This polynomial identity extends to $\ang{4}{5}=0$, where the uncancelled factor $\Delta_{345}^{-1}$ is singular. Proposition~\ref{prop:hodges-member} gives the remaining pair conditions for $H$.

\paragraph{A conserving family on the collinear boundary.}
Keep $X_4=X_5$ and vary $Y_5$, compensating with $Y_1,Y_2$. In column components, take
\begin{equation}
\begin{array}{c|ccccc}
 i&1&2&3&4&5\\ \hline
 X_i&(1,0)&(0,1)&(1,1)&(1,2)&(1,2)\\
 Y_i&(-3,-7-t)&(-5,-13-2t)&(1,1)&(1,3)&(1,3+t)
\end{array}
\label{eq:toy-spinors}
\end{equation}
Conservation holds for every $t$, since
\begin{equation}
 \sum_{i=3}^5X_iY_i^{\mathsf T}
 =\begin{pmatrix}3&7+t\\5&13+2t\end{pmatrix},\qquad
 X_1Y_1^{\mathsf T}+X_2Y_2^{\mathsf T}
 =-\begin{pmatrix}3&7+t\\5&13+2t\end{pmatrix}.
 \label{eq:toy-momentum}
\end{equation}
Here $\ang{4}{5}=0$ throughout, whereas $\sqr{4}{5}=t$. Hence $t=0$ is a simultaneous pair zero, and $t\ne0$ tests the residue on the angle-collinear branch. The relevant square brackets and row entries are
\begin{equation}
\begin{array}{c|cc|cc}
 i&[1i]&[2i]&\rho_i&\sigma_i\\ \hline
 4&t-2&2t-2&4t-4&2-t\\
 5&-2-2t&-2-3t&-4-6t&2+2t
\end{array}
\label{eq:toy-rows}
\end{equation}
For example, $[24]=(-5)3-(-13-2t)=2t-2$, while $\ang{1}{4}=\ang{1}{5}=2$ and $\ang{2}{4}=\ang{2}{5}=-1$. Multiplication now gives
\begin{equation}
 \begin{aligned}
 \rho_5\sigma_4&=-8-8t+6t^2,&
 \rho_4\sigma_5&=-8+8t^2,\\
 H\big|_{\eqref{eq:toy-spinors}}&=-2t(t+4).
 \end{aligned}
 \label{eq:toy-products}
\end{equation}
At $t=0$, both products are $-8$, so $N=-8(a+b)$. The pair zero forces $b=-a$, leaving $N=aH$ with arbitrary $a$.

\paragraph{The residue and its normalization.}
Factor the common denominator before restricting:
\begin{equation}
 D_5=\ang{4}{5}\widehat D_{45},\qquad
 \left.\ang{4}{5}\frac{N}{D_5}\right|_{\ang{4}{5}=0}
 =\left.\frac{N}{\widehat D_{45}}\right|_{\ang{4}{5}=0}.
 \label{eq:toy-residue-definition}
\end{equation}
In \eqref{eq:toy-spinors}, the remaining nine angle factors are
\begin{equation}
\begin{array}{c|rrrrrrrrr}
 ij&12&13&14&15&23&24&25&34&35\\ \hline
 \ang{i}{j}&1&1&2&2&-1&-1&-1&1&1
\end{array}
\qquad \widehat D_{45}=-4.
\label{eq:toy-denominator}
\end{equation}
Dividing by $\widehat D_{45}=-4$ gives the residue
\begin{equation}
 \left.\ang{4}{5}\frac{aH}{D_5}\right|_{\eqref{eq:toy-spinors}}
 =\frac{a[-2t(t+4)]}{-4}
 =\frac a2t(t+4).
 \label{eq:toy-residue-function}
\end{equation}
In particular, at $t=1$ the two products in \eqref{eq:toy-products} are $-10$ and $0$, and the residue is $5a/2$.

Compute the factorization input independently by fusing legs $4,5$ into $X_P=(1,2)^{\mathsf T}$, $Y_P=(2,6+t)^{\mathsf T}$. Delete legs $1,2,3$ from the four-point Hodges matrix and choose reference $q=X_1$. Since $\Delta_{123}=1$, its one-by-one reduced minor is
\begin{equation}
 \begin{aligned}
 \mathcal A_4=\Phi_{PP}
 &=-\left(
 \frac{[P2]\ang{2}{1}^2}{\ang{P}{2}\ang{P}{1}^2}
 +\frac{[P3]\ang{3}{1}^2}{\ang{P}{3}\ang{P}{1}^2}
 \right)\\
 &=-\left(\frac{t+4}{4}+\frac{t+4}{4}\right)
 =-\frac{t+4}{2}.
 \end{aligned}
 \label{eq:toy-four-point}
\end{equation}
The omitted $j=1$ term has $\ang{1}{q}=0$. The splitting law \eqref{eq:physical-splitting}, with angle-spinor ratio one, prescribes $-[45]\mathcal A_4=t(t+4)/2$. At $t=1$, matching $5a/2$ to $5/2$ fixes $a=1$: the pair zero selects the line, and the normalized residue fixes its scale. The script \texttt{check\_five\_point\_worked.py} verifies these identities, the minor sign and the residue. Larger row determinants at seven and eight points preserve pair zeros while adding directions beyond the Hodges line.

\section{From the Hodges formula to explicit counterexamples}
\label{sec:lowpoints}
In our conventions the Hodges numerator is
\begin{equation}
 \begin{aligned}
 H_n&=N_n^{\mathrm H}=D_n\det'\Phi_n,\qquad
 \Phi_{ij}=\frac{[ij]}{\ang{i}{j}}\quad(i\ne j),\\
 \Phi_{ii}&=-\sum_{j\ne i}\frac{[ij]\ang{j}{q}^2}{\ang{i}{j}\ang{i}{q}^2}.
 \end{aligned}
 \label{eq:counter-hodges-start}
\end{equation}
Here $q$ is a reference spinor. Deleting the same triple $a,b,c$ from the rows and columns gives $\det'\Phi_n=\det\bigl((\Phi_n)^{abc}_{abc}\bigr)/(\ang{a}{b}\ang{b}{c}\ang{c}{a})^2$; momentum conservation makes the answer independent of the deleted triple. Appendix~\ref{app:hodges-tech} establishes the polynomiality and pair zeros of $H_n$.

A polynomial correction $F_n$ with the prescribed degrees and all pair zeros gives an allowed family $H_n+\lambda F_n$. We establish its independence from $H_n$ at regular conserving kinematics where the Hodges minor vanishes and $F_n$ is nonzero.

Construct the corrections from entries of the Hodges matrix. Choose three distinct labels $a,b,c$, let $B$ be their increasing complement, and set $\varepsilon_{abcB}=\sgn(a,b,c,B)$. For an auxiliary spinor $q$ with $\ang{a}{q}\ne0$, define
\begin{equation}
 C_a^{(k)}(i)=\left.\frac{d^k}{d\tau^k}
 \frac{[ia]}{\ang{i}{a}+\tau\ang{i}{q}}\right|_{\tau=0}
 =(-1)^kk!\frac{[ia]\ang{i}{q}^{k}}{\ang{i}{a}^{k+1}},\qquad i\in B.
 \label{eq:hodges-confluent-columns}
\end{equation}
Thus $C_a^{(0)}(i)=\Phi_{ia}$. Coalescing two or three columns at the same label and retaining successive derivatives gives a confluent determinant. The derivatives hold $Y_a$ and all other spinors fixed, before imposing momentum conservation. The determinant has four rows at seven points and five at eight. Proposition~\ref{prop:hodges-confluent} proves that the prefactors below give globally defined polynomials independent of $b,c,q$.

\subsection{A regular configuration where the Hodges minor vanishes}
\label{subsec:hodges-rank-witness}
For $n=7,8$, set
\begin{equation}
 z_i=i-1,\qquad c_i=(-1)^{i-1}\binom{n-1}{i-1},\qquad
 X_i=(1,z_i)^{\mathsf T},\qquad Y_i=c_i(1,z_i^3)^{\mathsf T}.
 \label{eq:integer-cubic-witness}
\end{equation}
The table gives both configurations. Normalize determinant values by the nonzero common factor $P_cD_z$, where
\begin{equation}
 P_c=\prod_{i=1}^n c_i,\qquad
 D_z=\prod_{i<j}(z_j-z_i)=\prod_{j=1}^{n-1}j!.
 \label{eq:worked-normalization}
\end{equation}
\begin{center}
\small
\begin{tabular}{c@{\qquad}c@{\qquad}r@{\qquad}r}
\toprule
$i$ & $X_i^{\mathsf T}$ & $(Y_i^{(7)})^{\mathsf T}$ & $(Y_i^{(8)})^{\mathsf T}$\\
\midrule
1 & $(1,0)$ & $(1,0)$ & $(1,0)$\\
2 & $(1,1)$ & $(-6,-6)$ & $(-7,-7)$\\
3 & $(1,2)$ & $(15,120)$ & $(21,168)$\\
4 & $(1,3)$ & $(-20,-540)$ & $(-35,-945)$\\
5 & $(1,4)$ & $(15,960)$ & $(35,2240)$\\
6 & $(1,5)$ & $(-6,-750)$ & $(-21,-2625)$\\
7 & $(1,6)$ & $(1,216)$ & $(7,1512)$\\
8 & $(1,7)$ & --- & $(-1,-343)$\\
\bottomrule
\end{tabular}
\end{center}

Define $m_k=\sum_i c_i z_i^k$. Applying $(t\partial_t)^k$ to $(1-t)^{n-1}=\sum_i c_it^{z_i}$ and setting $t=1$ gives
\begin{equation}
 m_k=0\qquad(0\le k\le n-2).
 \label{eq:integer-moments}
\end{equation}
Indeed each application lowers the order of the zero at $t=1$ by at most one. The four momentum equations are exactly $m_0=m_1=m_3=m_4=0$, so these are conserving configurations. Moreover, all $\ang{i}{j}=z_j-z_i$ with $i\ne j$ are nonzero. Thus $D_n\ne0$, and these configurations lie away from poles and pair-zero loci.

Take $q=(0,1)^{\mathsf T}$. Off the diagonal the cubic difference factors as
\[
 \Phi_{ij}=c_ic_j\frac{z_j^3-z_i^3}{z_j-z_i}
          =c_ic_j(z_i^2+z_iz_j+z_j^2).
\]
The diagonal has the same formula. In fact the full row sum of this expression is $c_i(z_i^2m_0+z_im_1+m_2)=0$, as required by the Hodges diagonal. Hence
\begin{equation}
 \begin{aligned}
 \Phi_n&=G_cUA_0U^{\mathsf T}G_c,\qquad (U_i)=(1,z_i,z_i^2),\\
 A_0&=\begin{pmatrix}0&0&1\\0&1&0\\1&0&0\end{pmatrix},\qquad
 G_c=\operatorname{diag}(c_i).
 \end{aligned}
 \label{eq:counter-rank-three}
\end{equation}
The factorization has rank at most three, while the retained Hodges minors have sizes four and five. Thus $H_7=H_8=0$ at these regular configurations. Appendix~\ref{app:worked-witness} gives the retained integer matrices and their rank relations.

These regular amplitude zeros are not imposed by the pair conditions. The confluent columns can survive the rank drop because the derivative in \eqref{eq:hodges-confluent-columns} holds $Y_a$ fixed. Evaluating these derivatives at $q=(0,1)^{\mathsf T}$ on the cubic configuration gives
\[
 C_a^{(1)}(i)=-c_ic_a\frac{z_i^2+z_iz_a+z_a^2}{z_a-z_i},\qquad
 C_a^{(2)}(i)=2c_ic_a\frac{z_i^2+z_iz_a+z_a^2}{(z_a-z_i)^2}.
\]
These are not confined to the three polynomial columns of $U$. Differentiating the already specialized cubic factorization would change $Y_a$ as well and would compute a different derivative.

\subsection{Seven points: an explicit deformation without Bose symmetry}
\label{subsec:counter-seven}
For each marked label $a$, take the following four-by-four confluent minor of the off-diagonal Hodges entries:
\begin{equation}
 N_a=\varepsilon_{abcB}\,
 \frac{\prod_{i\in B}\ang{a}{i}^{2}\ang{b}{i}\ang{c}{i}}{\ang{a}{q}}
 \det_{i\in B}\bigl(\Phi_{ia},C_a^{(1)}(i),\Phi_{ib},\Phi_{ic}\bigr).
 \label{eq:seven-hodges-confluent}
\end{equation}
This gives the family
\begin{equation}
 F_7=N_1-N_2,\qquad N_7^{(\lambda)}=H_7+\lambda(N_1-N_2).
 \label{eq:seven-hodges-family}
\end{equation}
For polynomiality, degrees and zeros, use the equivalent determinant in Proposition~\ref{prop:hodges-confluent}. Write $m_p(x,y)=(x^p,x^{p-1}y,\ldots,y^p)$ and $r_i^{(p)}=(u_i m_p(X_i),v_i m_p(X_i))$. Then
\begin{equation}
 N_a=\frac{\det(r_1^{(3)},\ldots,r_7^{(3)},\partial_{y_a}r_a^{(3)})}{x_a}.
 \label{eq:seven-polynomial-certificate}
\end{equation}
The numerator of this expression is divisible by $x_a$: Euler's identity $x_ar_{a,x}+y_ar_{a,y}=3r_a$ gives $x_a\det(\ldots,r_{a,x})+y_a\det(\ldots,r_{a,y})=0$, and $x_a,y_a$ are coprime. The ordinary legs have spinor degrees $(3,1)$ and the marked leg $(4,2)$, giving total $(22,8)$ and weight two on every leg. The directional identity $\det(\ldots,D_\eta r_a)=\ang{X_a}{\eta}N_a$ shows Lorentz invariance. At any free pair zero two original rows are proportional; since the pair ideal is prime and does not contain $x_a$, division by $x_a$ preserves membership in it. Thus $F_7$ and every $N_7^{(\lambda)}$ have the required degrees, weights and pair zeros.

On the integer configuration~\eqref{eq:integer-cubic-witness}, the polynomial determinant gives
\begin{equation}
 \frac{N_a}{P_cD_z}=3c_az_a^2,\qquad
 \frac{(N_1,\ldots,N_7)}{P_cD_z}=(0,-18,180,-540,720,-450,108).
 \label{eq:seven-witness-short}
\end{equation}
To see the mechanism, write $Y=c(1,w)$ on $x=1$ and replace the column $w$ by $w-z^3$ before specializing to $w=z^3$. The ordinary rows reduce to a Vandermonde matrix, while the added derivative reads $-3z_a^2$ in the relation column. Appendix~\ref{app:worked-seven} gives the column signs, normalization factors and a direct evaluation from the Hodges entries. In particular,
\[
 H_7=0,\qquad \frac{F_7}{P_cD_z}=0-(-18)=18\ne0.
\]
Every $\lambda\ne0$ in~\eqref{eq:seven-hodges-family} therefore gives a numerator that is not proportional to $H_7$.

This seven-point freedom is removed by Bose symmetry. Relabelling gives $\sigma N_a=\sgn(\sigma)N_{\sigma(a)}$, so the transposition $(12)$ fixes $F_7$ but negates $D_7$. Thus $F_7/D_7$ is not Bose symmetric.

\subsection{Eight points: a Bose-symmetric deformation of Hodges}
\label{subsec:counter-eight}
Now retain the second derivative as well. The five-by-five confluent Hodges minor is
\begin{equation}
 \begin{aligned}
 R_a={}&\varepsilon_{abcB}\,
 \frac{\prod_{i\in B}\ang{a}{i}^{3}\ang{b}{i}\ang{c}{i}}{\ang{a}{q}^{3}}\\
 &\times\det_{i\in B}\bigl(\Phi_{ia},C_a^{(1)}(i),C_a^{(2)}(i),\Phi_{ib},\Phi_{ic}\bigr).
 \end{aligned}
 \label{eq:eight-hodges-confluent}
\end{equation}
Sum over the marked label to obtain a correction with the required permutation symmetry:
\begin{equation}
 \mathcal R=\sum_{a=1}^8R_a,\qquad
 \boxed{N_8^{(\lambda)}=D_8\det'\Phi_8+\lambda\mathcal R.}
 \label{eq:eight-hodges-family}
\end{equation}
The equivalent polynomial form is
\begin{equation}
 R_a=\frac{\det(r_1^{(4)},\ldots,r_8^{(4)},\partial_{y_a}r_a^{(4)},\partial_{y_a}^2r_a^{(4)})}{x_a^3}.
 \label{eq:eight-polynomial-certificate}
\end{equation}
These identities hold for arbitrary spinors. In a nonzero three-column minor of $(r_a,r_{a,y},r_{a,yy})$, the three powers $p,q,s$ of $y_a$ are distinct elements of $\{0,\ldots,4\}$. Its $x_a$ exponent is $12-p-q-s\ge3$, so the displayed quotient is polynomial. The marked leg has degrees $(6,3)$ and the ordinary legs $(4,1)$, giving the target total $(34,10)$ and weight three on every leg. Euler's identity also gives
\[
 \det(r_1,\ldots,r_8,D_\eta r_a,D_\eta^2r_a)=\ang{X_a}{\eta}^{3}R_a,
\]
which makes invariance independent of the chosen derivative direction. Each determinant retains all eight original rows, so the same proportional-row and prime-ideal argument proves every free pair condition. Finally, relabelling permutes the marks and contributes the row-permutation sign: $\mathcal R$ is alternating and $\mathcal R/D_8$ is Bose symmetric.

For the eight-point integer configuration, the same column-reduction argument leaves a Vandermonde block and the two-by-two derivative block
\[
 \det\begin{pmatrix}-3z_a^2&-3z_a^3\\-6z_a&-12z_a^2\end{pmatrix}=18z_a^4.
\]
The original rows contribute $P_cD_z$, and the two derivative rows contribute $c_a^2$. Hence
\begin{equation}
 \boxed{H_8=0,\qquad
 \frac{\mathcal R}{P_cD_z}=18\sum_{a=1}^8c_a^2z_a^4=13706280>0.}
 \label{eq:eight-counterexample}
\end{equation}
Appendix~\ref{app:worked-eight} gives the complete reduction and checks it directly against a marked confluent Hodges minor. Since $P_cD_z\ne0$, every $\lambda\ne0$ in~\eqref{eq:eight-hodges-family} preserves the prescribed degrees, all pair zeros and Bose symmetry while giving a numerator not proportional to $H_8$.

The corresponding stripped candidates are
\[
 \mathcal A_8^{(\lambda)}=\det'\Phi_8+\lambda\frac{\mathcal R}{D_8}.
\]
The polynomial construction proves global admissibility; the regular configuration establishes nonproportionality. The physical tests in Section~\ref{sec:physical} fix $\lambda=0$. The scripts \nolinkurl{check_hodges_deformation.py} and \nolinkurl{check_worked_counterexamples.py} verify the spinors, momentum equations, ranks, column reductions and table entries by exact arithmetic.

\subsection{Completeness in the two target sectors}
\label{subsec:counter-completeness}
For the classification at seven points, restrict the right spinors to two lines, three spinors on one line and four on the other. Each choice of the triple $P$ leaves a single coefficient $k_P$ after the forced same-line angle factors have been divided out. A double collision gives $k_{iac}-k_{ibc}-k_{iad}+k_{ibd}=0$ for distinct $i,a,b,c,d$, whence $k_P=\sum_{a\in P}c_a$. The restriction detects every class, and $k_P(N_a)=\mathbf1_{a\in P}$. Thus the seven coefficients are independent; equal coefficients give the alternating line.

For eight-point completeness we introduce a second alternating polynomial $\mathcal S$. Here $r_i=r_i^{(4)}$. Set $B_a=(r_a\wedge\partial_{y_a}r_a)/x_a$, and contract the derivatives of two marked blocks:
\begin{equation}
 \mathcal S=\sum_{a<b}(-1)^{a+b}\operatorname{vol}\!\left[
 \bigwedge_{i\notin\{a,b\}}^{\mathrm{increasing}}r_i\wedge
 (\partial_{x_a}B_a\wedge\partial_{y_b}B_b-\partial_{y_a}B_a\wedge\partial_{x_b}B_b)\right].
 \label{eq:counter-S-definition}
\end{equation}
The volume reads the ten-row determinant; Appendix~\ref{sec:eight} proves polynomiality and all pair conditions. A filtration of the free pair-restriction kernel gives at most four directions. Quotient classes may also restrict to momentum relations on the boundary. The exact sequence in Appendix~\ref{sec:framework} adds at most two such directions and subtracts four source relations, giving the alternating-sector bound $4+2-4=2$. The two independent generators $\mathcal R,\mathcal S$ attain it and classify the space containing \eqref{eq:eight-hodges-family}.
\begin{theorem}[Classification and lifting in the two target degrees]
\label{thm:main}
With the globally polynomial generators defined in Appendix~\ref{app:classification},
\begin{align}
 W_7&=\bigoplus_{a=1}^7\C\,[N_a],&\dim W_7&=7,\notag\\
 W_7^{\sgn}&=\C\,\Bigl[\sum_aN_a\Bigr],&
 W_8^{\sgn}&=\C\,[\mathcal R]\oplus\C\,[\mathcal S].
 \label{eq:main-spaces}
\end{align}
Every class in the seven-point space and the alternating eight-point space has a representative in \(I_n\).
In the determinant conventions used here,
\begin{equation}
 N_7^{\mathrm H}=-\frac13\sum_aN_a,
 \qquad N_8^{\mathrm H}=\frac{6\mathcal R+\mathcal S}{144}.
 \label{eq:main-hodges}
\end{equation}
\end{theorem}

\section{Physical uniqueness on the candidate space}
\label{sec:physical}
The eight-point counterexample leaves one coefficient ratio to determine. Write
\[
 \begin{gathered}
 H=N_8^{\mathrm H}=(6\mathcal R+\mathcal S)/144,\qquad
 N=\alpha H+\beta\mathcal R=c_R\mathcal R+c_S\mathcal S,\\
 \alpha=144c_S,\qquad \beta=c_R-6c_S.
 \end{gathered}
\]
The Einstein line is $\beta=0$; Lemma~\ref{lem:plane-test} tests its boundary coefficients. Our convention $\mathcal A_n=(-1)^{n+1}\widehat{\mathcal M}_n$ fixes the relative signs of the stripped Einstein amplitude in the collinear and soft formulas.

\paragraph{Conserving data for the worked examples.}
The examples use $X_i=(1,z_i)^{\mathsf T}$ and $Y_i=(u_i,v_i)^{\mathsf T}$ from Table~\ref{tab:physical-inputs}. All three right-spinor columns satisfy
\[
 \sum_i u_i=\sum_i z_iu_i=\sum_i v_i=\sum_i z_iv_i=0,
\]
and hence conserve momentum. Data $0$ give the large-parameter test; data $A,B$ give two independent boundary measurements.
\begin{table}[htbp]
\centering\small
\begin{tabular}{ccccc}
\toprule
$i$ & $z_i$ & $Y_i^{(0)}$ & $Y_i^{(A)}$ & $Y_i^{(B)}$\\
\midrule
1&0&$(1,0)$&$(11,6)$&$(14,11)$\\
2&1&$(-2,1)$&$(-13,-9)$&$(-19,-14)$\\
3&2&$(1,-2)$&$(1,2)$&$(2,1)$\\
4&3&$(1,1)$&$(-2,0)$&$(1,-1)$\\
5&4&$(-2,0)$&$(1,-1)$&$(-1,1)$\\
6&5&$(1,1)$&$(1,2)$&$(1,2)$\\
7&6&$(0,-2)$&$(-1,1)$&$(3,-1)$\\
8&7&$(0,1)$&$(2,-1)$&$(-1,1)$\\
\bottomrule
\end{tabular}
\caption{Integer conserving data. Pairs in the three rightmost columns denote column spinors.}
\label{tab:physical-inputs}
\end{table}

\subsection{Collinear factorization}
\label{subsec:collinear}
Bring two positive-helicity angle spinors together, $X_a=tX_b$, $t\ne0$, keeping $\sqr ab$ generic. The fused leg has $X_P=X_b$, $Y_P=tY_a+Y_b$; write $\iota_{ab,t}$ for its insertion in the lower-point function. The three-point amplitude $\mathcal M_3(a^+,b^+,(-P)^-)=\sqr ab^2/t^2$ and the propagator give the tree-level residue \cite{BernEtAl1998}
\begin{equation}
 \left.\ang ab\mathcal A_n\right|_{X_a=tX_b}
  =-t^{-2}\sqr ab\,\iota_{ab,t}(\mathcal A_{n-1}).
 \label{eq:physical-splitting}
\end{equation}
\paragraph{Two explicit eight-point residues.}
Starting from either data $A$ or data $B$, make the replacements
\begin{equation}
 X_7\longmapsto X_8,\qquad
 Y_1\longmapsto Y_1+Y_7,\qquad
 Y_2\longmapsto Y_2-Y_7.
 \label{eq:worked-collinear-path}
\end{equation}
Here $X_8-X_7=X_2-X_1=(0,1)^{\mathsf T}$ before the replacement. The momentum change is therefore $(X_2-X_1)Y_7^{\mathsf T}+X_1Y_7^{\mathsf T}-X_2Y_7^{\mathsf T}=0$. This gives the boundary $a=7,b=8,t=1$. The changed pairs $(Y_1,Y_2)$ are $((10,7),(-12,-10))$ for $A$ and $((17,10),(-22,-13))$ for $B$.

On this boundary the denominator with its vanishing factor removed is
\[
 \overline D_{78}=\left.\frac{D_8}{\ang78}\right|_{X_7=X_8}
 =877879296000,\qquad
 L(N)=\frac{N|_{X_7=X_8}}{\overline D_{78}}.
\]
The fused seven-point input is $(X_1,\ldots,X_6,X_8)$ with right spinors $(Y_1,\ldots,Y_6,Y_7+Y_8)$ after the replacements. Its Hodges minor gives Table~\ref{tab:worked-collinear}, where each row satisfies $L(H)=-[78]\mathcal A_7$.
\begin{table}[htbp]
\centering\small\renewcommand{\arraystretch}{1.35}
\begin{tabular}{ccccc}
\toprule
Data & $[78]$ & $\mathcal A_7$ & $L(H)$ & $r=L(\mathcal R)/L(H)$\\
\midrule
$A$&$-1$&$\frac{2}{75}$&$\frac{2}{75}$&$-\frac{2551313}{3456}$\\
$B$&$2$&$\frac{347}{2688}$&$-\frac{347}{1344}$&$-\frac{1540883}{34700}$\\
\bottomrule
\end{tabular}
\caption{Two collinear residues, with the lower-point amplitude computed independently.}
\label{tab:worked-collinear}
\end{table}

Agreement with the physical residues up to one common scale $\gamma$ gives
\begin{equation}
 \begin{pmatrix}1&r_A\\1&r_B\end{pmatrix}
 \begin{pmatrix}\alpha-\gamma\\\beta\end{pmatrix}=0.
 \label{eq:worked-boundary-system}
\end{equation}
Both $L(H)$ are nonzero, and the system has determinant $r_B-r_A\ne0$. Thus $\beta=0$ and $\alpha=\gamma$, with $\gamma=1$ for the normalized seven-point input. These two evaluations determine the candidate within the classified eight-point plane; complete boundary functions give the following statement at arbitrary multiplicity.

\begin{proposition}[Complete collinear data]\label{thm:boundary-unique}\label{thm:physical-unique}
For $n\ge5$, fix the bracket degrees $(A,B)$ and leg weights with $A<\binom n2$. The full restrictions to one fixed nonzero ratio for each pair are jointly injective. A single pair suffices for alternating numerators with equal leg weights. Here a full restriction is an identity as a function of the remaining kinematics, rather than evaluation at one configuration. In the MHV degree the normalized physical boundary in \eqref{eq:physical-splitting} has the unique extension $N_n^{\mathrm H}$ for a lower-point Hodges input.
\end{proposition}
\begin{proof}
For the difference of two candidates, vanishing at one ratio as a function of the other spinors extends by little-group rescaling to all nonzero ratios. These points are dense on the prime divisor $\ang ab=0$ (Proposition~\ref{prop:angle-prime}). Relabelling supplies the other pairs in the alternating case. Each of the distinct prime angle factors divides the difference, so their product $D_n$ does too. The remaining angle degree is negative. To pass from the residue to the numerator restriction, multiply by $D_n/\ang ab$ on the generic branch. Proposition~\ref{prop:hodges-splitting} supplies the physical extension.
\end{proof}

Equal complete residues remove every allowed pole, and the prescribed degree excludes a contact term. The proof establishes the required divisibility in the momentum-constrained invariant ring. Bose symmetry reduces the data to one full channel; the splitting law and a normalized lower-point amplitude supply its value.

For example, with $a=n-1,b=n$ and $N_n=c_n N_n^{\mathrm H}$, the full fused numerator at $n\ge6$ is
\begin{equation}
 \left.N_n\right|_{X_a=tX_b}
 =-c_n t^{n-4}\sqr ab
   \left(\prod_{i=1}^{n-2}\ang ib\right)\iota_{ab,t}(N_{n-1}^{\mathrm H}).
 \label{eq:fused-numerator}
\end{equation}
Once the line is known, one nonzero residue evaluation gives $c_n=c_{n-1}$. The normalized five-point seed in Section~\ref{subsec:toy} is
\begin{equation}
 N_5^{\mathrm H}=\ang15\ang24\sqr14\sqr25-\ang14\ang25\sqr15\sqr24 .
 \label{eq:five-seed}
\end{equation}

\subsection{The leading soft coefficient}
\label{subsec:soft}
Take $X_s\mapsto\epsilon X_s$ for a positive-helicity leg $s$ and adjust two hard right spinors to preserve momentum. At fixed reference spinor $q$, the physical leading soft theorem in our phase convention is \cite{CachazoStrominger2014}
\begin{equation}
 \lim_{\epsilon\to0}\epsilon^3\mathcal A_n(\epsilon)
 =-S_s^{(0)}\mathcal A_{n-1},\qquad
 S_s^{(0)}=\sum_{i\ne s}\frac{\sqr si}{\ang si}
       \frac{\ang qi^2}{\ang qs^2}.
 \label{eq:leading-soft-test}
\end{equation}
The right side uses the hard spinors at $\epsilon=0$ and the unscaled soft direction $X_s$. Every candidate on the eight-point plane has $N=O(\epsilon^4)$ and $D_8=O(\epsilon^7)$, so pole order permits the entire plane. The leading coefficient distinguishes its directions.

\paragraph{A conserving soft path and its coefficient.}
Use data $A,B$ from Table~\ref{tab:physical-inputs}, put $s=8$, and set
\begin{equation}
 \begin{aligned}
 X_8(\epsilon)&=\epsilon X_8,\\
 Y_1(\epsilon)&=Y_1+6(\epsilon-1)Y_8,\qquad
 Y_2(\epsilon)=Y_2-7(\epsilon-1)Y_8.
 \end{aligned}
 \label{eq:worked-soft-path}
\end{equation}
All other spinors are fixed. The identity $X_8=-6X_1+7X_2$ makes the total momentum change vanish for every $\epsilon$. At $\epsilon=0$, the changed hard pairs $(Y_1,Y_2)$ are $((-1,12),(1,-16))$ for $A$ and $((20,5),(-26,-7))$ for $B$. Use these shifted spinors in the seven-point minor.

Exactly seven angle factors contain leg $8$, so
\[
 D_8(\epsilon)=\epsilon^7D_8(1),\qquad
 D_8(1)=125411328000,\qquad
 L(N)=\frac{[\epsilon^4]N(\epsilon)}{D_8(1)}.
\]
Substitute the path in the determinants defining $\mathcal R,\mathcal S$. The largest sum of degrees on $X_8,Y_1,Y_2$ is eight, so evaluations at $\epsilon=1,\ldots,9$ reconstruct each polynomial exactly. Coefficients below degree four vanish; the degree-four coefficient gives $L$. An independent soft-factor calculation with $q=X_1$ gives Table~\ref{tab:worked-soft}.
\begin{table}[htbp]
\centering\small\renewcommand{\arraystretch}{1.35}
\begin{tabular}{ccccc}
\toprule
Data & $S_8^{(0)}$ & $\mathcal A_7$ & $L(H)=-S_8^{(0)}\mathcal A_7$ & $r=L(\mathcal R)/L(H)$\\
\midrule
$A$&$-\frac{25}{14}$&$-\frac{2281}{960}$&$-\frac{11405}{2688}$&$\frac{1145234}{57025}$\\
$B$&$\frac{76}{35}$&$-\frac{7829}{1280}$&$\frac{148751}{11200}$&$\frac{26194735}{5355036}$\\
\bottomrule
\end{tabular}
\caption{Leading soft coefficients on the two specified conserving paths.}
\label{tab:worked-soft}
\end{table}

The distinct ratios in Table~\ref{tab:worked-soft}, inserted in~\eqref{eq:worked-boundary-system}, force $\beta=0$. The coefficient prescribed by normalized seven-point input fixes $\alpha=1$.

\subsection{Large-parameter BCFW decay}
\label{subsec:shift}
Use data $0$ from Table~\ref{tab:physical-inputs} and shift
\[
 Y_1(Z)=Y_1+ZY_2=(1-2Z,Z)^{\mathsf T},\qquad
 X_2(Z)=X_2-ZX_1=(1-Z,1)^{\mathsf T}.
\]
The two momentum changes cancel exactly: $X_1(ZY_2)^{\mathsf T}+(-ZX_1)Y_2^{\mathsf T}=0$. Also $\ang12=1$ remains fixed. For two negative-helicity shifted legs, the stripped factor $\ang12^8$ is consequently unchanged.

Only the six brackets joining leg $2$ to spectators depend on $Z$, so $D_8(Z)$ has degree six. With $C=D_8(0)/6!=174182400$, exact determinant interpolation gives
\begin{equation}
 [Z^7](\mathcal R,\mathcal S)=(0,0),\qquad
 \frac{[Z^6](\mathcal R,\mathcal S,D_8)}{C}=(-317,1902,5040).
 \label{eq:shift-leading-short}
\end{equation}
Appendix~\ref{app:worked-shift} supplies the interpolation formula, its degree bound and the subleading coefficients. The ratio of the degree-six terms gives

\begin{equation}
 \lim_{Z\to\infty}\frac{N(Z)}{D_8(Z)}
 =-\frac{317}{5040}\,\beta=-\frac{317}{5040}(c_R-6c_S).
 \label{eq:eight-ordinary-test}
\end{equation}
A generic $O(Z^{-1})$ falloff forces the forbidden leading coefficients to vanish as functions of the conserving spinors. They must therefore vanish at this point as well. Equation~\eqref{eq:eight-ordinary-test} gives $c_R=6c_S$, or $\beta=0$; for the deformation $H+\lambda\mathcal R$, this is $\lambda=0$. The condition is homogeneous and leaves $\alpha$ free.

The degree-five entries in Table~\ref{tab:worked-shift} also cancel in $6\mathcal R+\mathcal S$, and direct extraction gives
\[
 H(Z)=C\left(-\frac{137}{672}Z^4+O(Z^3)\right).
\]
The surviving line therefore has $O(Z^{-2})$ decay in this example. The script \nolinkurl{check_physical_worked.py} reproduces all three tests and the independent lower-point coefficients.

\begin{theorem}[One-shift criterion]\label{thm:bcfw}
For $n\ge5$ and $N_n\in V_n^{\sgn}$, require $N_n(Z)/D_n(Z)=O(Z^{-1})$ under one fixed ordered shift \eqref{eq:bcfw-shift} at generic complex conserving kinematics. The solution space is $\C N_n^{\mathrm H}$. Its elements have $O(Z^{-2})$ decay.
\end{theorem}
\begin{corollary}\label{cor:bonus-all}
On $V_n$, imposing $N_n(Z)/D_n(Z)=O(Z^{-2})$ on generic kinematics for one orientation of every unordered pair also selects $\C N_n^{\mathrm H}$. These conditions concern the stripped amplitude $N_n/D_n$.
\end{corollary}
Appendix~\ref{sec:bcfw} proves both statements in the full fixed-degree ansatz. Each multiplicity retains an independent constant under the shift condition.

\subsection{Comparison of the three tests}
\label{sec:discussion}
\begin{table}[ht]
\centering\small
\begin{tabular}{@{}>{\raggedright\arraybackslash}p{.14\linewidth}>{\raggedright\arraybackslash}p{.265\linewidth}>{\raggedright\arraybackslash}p{.235\linewidth}>{\raggedright\arraybackslash}p{.265\linewidth}@{}}
\toprule
Test & Physical input & On the eight-point plane & Scope of the uniqueness argument\\
\midrule
Collinear & Residue on $\langle ab\rangle=0$ supplied by splitting and the lower-point amplitude. &
Two residue ratios distinguish $H$ from $\mathcal R$. &
For $n\ge5$, complete residue data on one pair determine at most one element of $V_n^{\sgn}$; all pairs suffice on $V_n$.\\[4pt]
Leading soft & Coefficient of $\epsilon^{-3}$ from the soft factor and the lower-point amplitude. &
Two soft-coefficient ratios distinguish $H$ from $\mathcal R$. &
The physical soft law holds at general $n$. We prove candidate selection on $W_8^{\sgn}$.\\[4pt]
BCFW & $O(1/Z)$ on a generic same-helicity shift. No lower-point amplitude is needed. &
The constant term at one shifted configuration measures $c_R-6c_S$. &
For $n\ge5$, one fixed shift on generic kinematics selects a line in $V_n^{\sgn}$. All-pair $O(Z^{-2})$ of $N_n/D_n$ selects it in $V_n$.\\
\bottomrule
\end{tabular}
\caption{All three tests give $c_R=6c_S$ in $N_8=c_R\mathcal R+c_S\mathcal S$. Collinear and soft inputs prescribed with an absolute normalization also fix the scale on this plane. A homogeneous shift condition leaves it free.}
\label{tab:inputs}
\end{table}

Theorem~\ref{thm:main} classifies the full seven-point pair-zero space and the Bose-symmetric eight-point sector. The unrestricted eight-point dimension is not determined here. At nine points the rank check in Appendix~\ref{app:checks} gives $1\le\dim W_9^{\sgn}\le5$.

\paragraph{Relation to other reconstructions.}
The Segre determinants of Koefler, Ventura and Yu give non-Hodges solutions at every $n\ge7$~\cite{KoeflerVenturaYu2026}; their Hodges map describes the physical amplitude divisor~\cite{Pentads}. Our deformations in Section~\ref{sec:lowpoints} give explicit distinct rational functions with the fixed denominator, with complete bases in the two classified sectors.

The reconstruction of~\cite{AHRT} derives factorization from gauge invariance, cubic-graph poles and minimal power counting. Our collinear and soft tests use coefficient functions supplied by Einstein gravity~\cite{BernEtAl1998,CachazoStrominger2014}. Gravity recursion obtains finite-pole residues from tree factorization~\cite{CachazoSvrcek2005,BedfordEtAl2005,BCFW,BenincasaGravity}, and large-complex-momentum analyses establish physical amplitude decay~\cite{ArkaniHamedKaplan2008}. The one-shift theorem selects the Hodges line by decay at infinity, independently of lower-point amplitudes or factorization residues. Related uses of shifts include the locality criterion of~\cite{RodinaUV}, Bose bonus scaling~\cite{McGadyRodina}, and gravity recursion with collinear zeros~\cite{ParanjapeTrnka}.

\paragraph{Other physical constraints.}
Leading soft coefficients suffice on the eight-point plane; scattering equations also yield subleading data~\cite{SchwabVolovich2014}. Pole order alone leaves the plane intact. Helicity zeros and smooth splitting~\cite{JonesParanjape2025} probe boundary information different from the normalized coefficient functions used here; the first version gives the corresponding eight-point checks~\cite{MaiZhangV1}.

Celestial Ward identities provide another route to the Hodges formula~\cite{GuevaraHimwichMiller2025}. Labelled trees and scattering equations organize other representations~\cite{FengZhang2020,FengHuZhang2022}, while positive geometry characterizes canonical forms through geometric boundary data~\cite{ArkaniHamedBaiLam2017}. Each construction supplies its own defining data; the present comparison concerns the common-denominator MHV space of Section~\ref{sec:setup}.
\acknowledgments
This work was supported by the National Natural Science Foundation of China (Grant No.~12605132), the Ningxia Natural Science Foundation (Grant No.~2026AAC030080), and the Research Start-up Fund for Young Faculty of Ningxia University.

\paragraph{Code and computational assistance.}
The ancillary files contain verification scripts and the nine-point certificate. Earlier computational material and the finite-dimensional Lean companion are in the \href{https://github.com/ybzhang-nxu/autoNMHV}{autoNMHV repository} and the first arXiv version~\cite{MaiZhangV1}; Appendix~\ref{subsec:verification-scope} describes the formal checks.

Machine-learning methods have been applied to bootstrap coefficients, compact spinor-helicity expressions and relations inferred from amplitude data~\cite{CaiEtAl2024Transforming,CheungDersySchwartz2025,Shih2026Unscramble,Moynihan2026SMatrix}. OpenAI Codex assisted with the exploration and drafting of the large-shift argument, implementation of the arithmetic checks, literature organization and manuscript editing.

\appendix
\setcounter{secnumdepth}{3}
\section{Algebraic preliminaries}
\label{app:algebra}
\subsection{Momentum ideals and soft representatives}
\label{sec:foundations}
Write $R_n=S_n/\mathfrak M_n$ and $Q_n=R_n^G$ as in Section~\ref{sec:setup}. We use the complete-intersection, depth and invariant-projection criteria of \cite{StacksProject,Procesi2007}.
\subsubsection{Prime divisors and successive division}
\begin{lemma}\label{lem:moment-height}
The heights $h_k=\operatorname{ht}\mathfrak M_k$ for $k=0,1,2$ are $0,2,3$, and $h_k=4$ for $k\ge3$.
\end{lemma}
\begin{proof}
On the rank-$r$ stratum of the $2\times k$ angle matrix $X$, the dimensions of $X$ and of the allowed $Y$ are $r(k+2-r)$ and $2(k-r)$. Minimize the codimension $(2-r)k+r^2$ over $0\le r\le\min(2,k)$.
\end{proof}
\begin{proposition}\label{prop:moment-domain}
For $n\ge4$, the momentum equations define a complete intersection and $R_n,Q_n$ are domains.
\end{proposition}
\begin{proof}
Over the irreducible rank-two $X$ stratum the solutions form a vector bundle of dimension $4n-4$, smooth since two independent columns solve for two $Y$'s. The lower strata have dimensions $3n-1,2n<4n-4$. By the height theorem every component meets the bundle; hence there is a unique component. The four generators have height four in a polynomial ring, so the quotient is Cohen--Macaulay. A generically reduced ring with no embedded primes is reduced. The invariant subring is also a domain.
\end{proof}
\begin{proposition}\label{prop:angle-prime}
For $n\ge5$ every angle bracket is prime in $R_n$ and $Q_n$. Distinct unordered pairs give nonassociate primes. The same holds for square brackets.
\end{proposition}
\begin{proof}
Adding $\ang ab=0$ replaces the rank-two base by an open set of an irreducible determinantal hypersurface. The three strata have dimensions $4n-5,3n-1,2n$. The lower two are too small to contain a component. At general proportional nonzero $X_a,X_b$, $d\ang ab$ contributes an independent differential in the $X$ variables. Height five and generic smoothness give a domain as above. Reynolds projection contracts this invariant principal ideal to $(\ang ab)Q_n$. Distinct minors have the same bidegree and are not constant multiples, as is seen at $Y=0$.
\end{proof}
\begin{proposition}\label{prop:pair-prime}
For $n\ge5$, imposing $\ang ab=\sqr ab=0$ gives domains in both $R_n$ and $Q_n$. No other angle or square minor lies in this ideal.
\end{proposition}
\begin{proof}
Let $O$ denote the ordinary $n-2$ columns. Over $\rank X_O=2$, impose the two independent determinantal conditions on $X$ and $Y_a,Y_b$, then solve for two ordinary $Y$'s. This gives an affine bundle over an irreducible base of total dimension $4n-6$. The remaining dimensions are
\[
\begin{array}{c|rrrr}
(\rank X_O,\rank X)&(1,2)&(1,1)&(0,1)&(0,0)\\\hline
\dim X&n+2&n+1&3&0\\
\dim Y&2n-4&2n-3&2n-2&2n-1
\end{array}
\]
For $(1,2)$ each angle line contributes two independent linear conditions on $Y$ and the marked determinant vanishes automatically. For $(1,1)$ one ordinary column solves momentum leaving the marked determinant. For $(0,1)$ only the marked linear equations survive; for $(0,0)$ only the marked determinant survives. All totals are below $4n-6$. At a general point the four ordinary-$Y$ momentum differentials and the two marked determinantal differentials are independent. The height-six complete intersection is irreducible and reduced. Reynolds projection contracts the two-generator ideal. Taking $Y=0$ and otherwise generic $X$ on $\ang ab=0$ excludes any other angle bracket; interchange $X,Y$ for the square brackets.
\end{proof}
\begin{corollary}\label{cor:all-angle-divisibility}
If $f\in Q_n$ is divisible by every angle bracket, then $D_n\mid f$. In angle degree $<\binom n2$ this implies $f=0$.
\end{corollary}
\begin{corollary}\label{cor:angle-denominator-clearing}
If $f$ has only angle denominators and $af\in(Q_n)_{(a)}$ for every angle factor $a$, then $D_nf\in Q_n$.
\end{corollary}
\begin{proof}
In the domain $Q_n$, $sF\in(a^m)$ and $s\notin(a)$ imply $a^m\mid F$ by induction and primality. Thus each power contracts from its prime localization. Nonassociate primes can be removed successively from any common denominator of $D_nf$, proving both assertions.
\end{proof}
\begin{lemma}\label{lem:free-pair-prime}
The free pair ideal $P_{ij}$ is prime. It contains no individual spinor coordinate or minor on another pair.
\end{lemma}
\begin{proof}
Its variety is the product of two independent reduced, irreducible rank-one cones and an affine space. Proportional marked columns leave the other minors generically nonzero.
\end{proof}
\begin{lemma}\label{lem:graph-divisibility}
In the free $X$ ring, divisibility by the brackets on a graph implies divisibility by their product. At vertex $i$ the residual degree drops by the valence.
\end{lemma}
\begin{proof}
These minors are pairwise nonassociate irreducibles in a polynomial ring.
\end{proof}
\subsubsection{Weighted momentum relations and lifting}
\begin{theorem}\label{thm:weighted-koszul}
Set $L_i=(u_i,v_i)$, $K_{\mathbf d}=\prod_i L_i^{d_i}$ with $d_i\ge1$. For $n\ge5$ the momentum equations are a regular sequence on $S_n/K_{\mathbf d}$. Their Koszul complex on $K_{\mathbf d}$ is exact in positive degrees, and
\begin{equation}
 K_{\mathbf d}\cap\mathfrak M_n=K_{\mathbf d}\mathfrak M_n.
 \label{eq:weighted-intersection}
\end{equation}
\end{theorem}
\begin{proof}
At a prime $\mathfrak p\supseteq K_{\mathbf d}+\mathfrak M_n$, suppose exactly $m\ge1$ of the $L_i$ are contained in $\mathfrak p$, and put $k=n-m$. The remaining factors become units. Each $(u,v)^d$ has a length-one resolution with generators $e_j=u^{d-j}v^j$ and independent adjacent relations $ve_j-ue_{j+1}$, by successively reducing modulo $v$. Tensor these resolutions in disjoint variables and append the inclusion of the product into $S_n$. The resulting localized resolution of the quotient is minimal, since its entries lie in $\mathfrak p$. Auslander--Buchsbaum gives
\[
\operatorname{depth}(S_n/K_{\mathbf d})_{\mathfrak p}
=\operatorname{ht}\mathfrak p-m-1\ \ge\ m+h_k-1.
\]
Indeed the $m$ zero spinors cost $2m$ coordinates and leave the $k$-leg momentum equations. The last bound is $m+3$ for $k\ge3$, and $n,n,n-1$ for $k=2,1,0$. The grade is therefore four. The four homogeneous generators form a regular sequence, also on $S_n$. The Koszul sequence of $0\to K_{\mathbf d}\to S_n\to S_n/K_{\mathbf d}\to0$ gives both positive-degree vanishing and injectivity of $K_{\mathbf d}/\mathfrak M_nK_{\mathbf d}\to R_n$.
\end{proof}
\begin{theorem}\label{thm:soft-lift}
For $n\ge6$ and $K=\prod_i L_i$, one has
\begin{equation}
 K+\mathfrak M_n=\bigcap_i(L_i+\mathfrak M_n).
 \label{eq:soft-union}
\end{equation}
Every pair-zero class of fixed weights $\chi_i<n-1$ has a representative in $K^G$ of the same degree and weights. When all weights agree, the permutation type may also be fixed.
\end{theorem}
\begin{proof}
After the four regular equations the preceding depth equals $\operatorname{ht}\mathfrak p-m-5\ge m+h_k-5$, positive for $m\ge2,n\ge6$. For $m=1$ an associated prime must have height six and contain $B_i=L_i+\mathfrak M_n$, itself prime of height six because $S_n/B_i=R_{n-1}[x_i,y_i]$. These $B_i$ are precisely the minimal components of the set-theoretic union. At each $B_i$ the other $L_j$ become units, so the branch is reduced there. No embedded primes remain, proving \eqref{eq:soft-union}.
At $Y_i=0$ every pair condition makes the restriction divisible by $\ang ij$ for $j\ne i$. Over $\operatorname{Frac}(R_{n-1})$ these are $n-1$ distinct linear factors in $X_i$ because the surviving $X_j$ are generically nonparallel. The $X_i$ degree $\chi_i<n-1$ forces zero. Thus the representative lies in every $B_i$; use \eqref{eq:soft-union} and project to the fixed degree, weights and $G$-invariants. In a permutation-stable degree the isotypic projector preserves these conditions. In the MHV degree the resulting multihomogeneous terms obey
\begin{equation}\label{eq:marking-budget}
\deg_{Y_i}=1+m_i,\quad \deg_{X_i}=n-4+m_i,\quad m_i\ge0,\quad\sum_i m_i=n-6.
\end{equation}
\end{proof}

\subsection{Three interpolation identities}
\label{sec:framework}
For a boundary restriction $\beta:\mathcal A\to\mathcal D$, write $\mathcal J,\mathcal H$ for source and boundary relations, respectively, with $\beta(\mathcal J)\subseteq\mathcal H$.
\begin{proposition}\label{thm:boundary-exact}
For $\mathcal U=\beta^{-1}(\mathcal H)/\mathcal J$, the following sequence is exact:
\begin{equation}
 0\longrightarrow\frac{\ker\beta}{\ker\beta\cap\mathcal J}
 \longrightarrow\mathcal U
 \longrightarrow\frac{\operatorname{im}\beta\cap\mathcal H}{\beta(\mathcal J)}
 \longrightarrow0.
 \label{eq:boundary-exact-general}
\end{equation}
\end{proposition}
\begin{proof}
The last map sends $[F]$ to $[\beta F]$. It is onto and its kernel consists of classes $[F-J]$ with $F-J\in\ker\beta$, $J\in\mathcal J$. Taking dimensions gives
\begin{equation}
 \dim\mathcal U=\dim\ker\beta+
 \dim(\operatorname{im}\beta\cap\mathcal H)-\dim\mathcal J
 \le\dim\ker\beta+\dim\mathcal H-\dim\mathcal J.\qedhere
 \label{eq:boundary-dimension-general}
\end{equation}
\end{proof}
\begin{proposition}\label{thm:subset-interpolation}
For $2\le r\le n-2$, data $k_P$ on $r$-subsets satisfy
\begin{equation}
 k_{Aac}-k_{Abc}-k_{Aad}+k_{Abd}=0
 \label{eq:subset-four-term}
\end{equation}
for every $(r-2)$-subset $A$ and distinct $a,b,c,d$ outside it if and only if there are unique $c_i$ with
\begin{equation}
 k_P=\sum_{i\in P}c_i.
 \label{eq:subset-additive}
\end{equation}
Writing $s_i=\sum_{P\ni i}k_P$ and $s=\sum_Pk_P$, one has
\begin{equation}
 c_i=\frac{s_i-\frac{r-1}{n-1}s}{\binom{n-2}{r-1}}.
 \label{eq:subset-inverse}
\end{equation}
\end{proposition}
\begin{proof}
For fixed $a,b$, the difference $k_{Ia}-k_{Ib}=q_{ab}$ is unchanged by one-label exchanges in the $(r-1)$-set $I$. Connectivity gives independence of $I$. Taking $I$ outside $a,b,c$ gives $q_{ab}+q_{bc}=q_{ac}$. Thus $q_{ab}=c_a-c_b$ and $k_P-\sum_{i\in P}c_i$ is constant. Absorb this constant by an equal shift of all $c_i$. Conversely the four-term expression cancels. Counting subsets containing one or two fixed labels gives $s_i=\binom{n-2}{r-1}c_i+\binom{n-2}{r-2}\sum_jc_j$ and $s=\binom{n-1}{r-1}\sum_jc_j$, proving the inverse and uniqueness.
\end{proof}
\begin{corollary}\label{cor:subset-completeness}
If a linear injection of a candidate space into the four-term data contains each incidence vector $P\mapsto\mathbf1_{i\in P}$ in its image, their preimages form a basis by \eqref{eq:subset-additive}.
\end{corollary}
\begin{lemma}\label{lem:triangle-right-inverse}
For symmetric edge data $R_{ab}$ on $n\ge5$ vertices put $S_a=\sum_{b\ne a}R_{ab}$ and $T=\sum_{a<b}R_{ab}$. Then $\sum_{c\ne a,b}H_{abc}=R_{ab}$ for
\begin{equation}
 H_{abc}=\frac{R_{ab}+R_{ac}+R_{bc}}{n-4}
 -\frac{S_a+S_b+S_c}{(n-4)(n-3)}
 +\frac{2T}{(n-4)(n-3)(n-2)}
 \label{eq:triangle-right-inverse}
\end{equation}
\end{lemma}
\begin{proof}
Summing the first numerator gives $(n-4)R_{ab}+S_a+S_b$, and summing the second gives $(n-3)(S_a+S_b)+2T$. Substitution proves the claim.
\end{proof}

\section{Classification at seven and eight points}
\label{app:classification}
\subsection{Seven-point completeness}
\label{sec:seven-v3}\label{sec:seven}
At $\tau_7=(22,8;2^7)$, choose the ordered rows and marked determinants
\begin{equation}
 r_i=(x_i^3u_i,x_i^2y_i u_i,x_i y_i^2u_i,y_i^3u_i,
       x_i^3v_i,x_i^2y_i v_i,x_i y_i^2v_i,y_i^3v_i),\qquad
 N_a=\frac{\det(r_1,\ldots,r_7,\partial_{y_a}r_a)}{x_a}.
 \label{eq:seven-v3-generators}
\end{equation}
\begin{theorem}\label{thm:seven-v3-main}
In the momentum quotient the seven classes $[N_a]$ form a basis of $W_7$. Each $N_a$ belongs to every free pair ideal $P_{ij}$.
\end{theorem}
\subsubsection{Detecting a momentum class with two square-spinor directions}
\label{subsec:seven-v3-ward}\label{subsec:seven-v3-detection}
By Theorem~\ref{thm:soft-lift}, choose $F\in(\prod_iL_i)^G_{\tau_7}$. Equation~\eqref{eq:marking-budget} gives a unique decomposition of this free representative:
\begin{equation}
 F=\sum_{a=1}^7F_a,\qquad
 (\deg_{X_i},\deg_{Y_i})F_a=
 \begin{cases}(4,2),&i=a,\\(3,1),&i\ne a.
 \end{cases}
 \label{eq:seven-v3-marked}
\end{equation}
Each $F_a$ is separately invariant. Put $E=\sum_i u_i\partial_{v_i}$, $L=\sum_i v_i\partial_{u_i}$ and $H=[E,L]$. Write $F_a=u_a^2 f_{a0}+u_av_af_{a1}+v_a^2f_{a2}$.
\begin{lemma}\label{lem:seven-v3-sl2}
For $f_{a0}$ of ordinary right weight $-2$, right invariance is equivalent to
\begin{equation}\label{eq:seven-v3-reconstruction}
L_{\ne a}f_{a0}=0,\qquad f_{a1}=-E_{\ne a}f_{a0},\qquad f_{a2}=E_{\ne a}^2 f_{a0}/2 .
\end{equation}
\end{lemma}
\begin{proof}
Compare coefficients in $EF_a=LF_a=0$. Conversely, the commutators give $LE^j g=j(m-j+1)E^{j-1}g$ for $Lg=0,Hg=-mg$. Finite-dimensionality implies $E^{m+1}g=0$ (apply the identity just past the last nonzero power). For $m=2$ these identities give both invariant equations.
\end{proof}
For a triple $P\ni a$, let $C_a^P$ denote the coefficient of $\prod_{i\in P-\{a\}}u_i\prod_{j\notin P}v_j$ in $f_{a0}$. These are all its monomials; lowering gives the relations
\begin{equation}
 \sum_{c\notin\{a,p\}}C_a^{\{a,p,c\}}=0
 \qquad(a\ne p).
 \label{eq:seven-v3-Ward}
\end{equation}
Define the alternating triple coordinate
\begin{equation}
 E_{abc}=\ang{b}{c}C_a^P+\ang{c}{a}C_b^P+\ang{a}{b}C_c^P,
 \qquad P=\{a,b,c\}.
 \label{eq:seven-v3-detector}
\end{equation}
Take $Y_i=t_i(1,0)$ on $P$ and $Y_j=t_j(0,1)$ on $Q=P^c$. Weight zero forces the mark into $P$. Schouten gives $(t_a,t_b,t_c)=\tau(\ang bc,\ang ca,\ang ab)$, and the ordinary four-column kernel on $Q$ permits general nonzero $t_j$. Hence
\begin{equation}
 F\big|_{\text{two colors, momentum}}=
 \tau\Bigl(\prod_{i=1}^7t_i\Bigr)E_{abc}.
 \label{eq:seven-v3-geometric-detector}
\end{equation}
\begin{theorem}\label{thm:seven-v3-detection}
All $E_P$ vanish identically if and only if $F\in\mathfrak M_7$.
\end{theorem}
\begin{proof}
A momentum relation has zero detector by the preceding generic specialization. Conversely, define
\begin{equation}
 H_P=\frac{C_a^PX_b-C_b^PX_a}{\ang{a}{b}}.
 \label{eq:seven-v3-triple-lift}
\end{equation}
The equation $E_{abc}=0$ gives $\ang ac(C_a^P X_b-C_b^P X_a)=\ang ab(C_a^P X_c-C_c^P X_a)$. Coprimality of the two minors in the free $X$ ring makes $H_P$ polynomial. It is the unique spinor with $\ang{X_i}{H_P}=C_i^P$ for $i\in P$. Equation~\eqref{eq:seven-v3-Ward} makes both contractions of $\sum_c H_{apc}$ with $X_a,X_p$ zero; thus this sum vanishes. Let
\[
 H^-=\sum_{|P|=3}H_P\prod_{i\in P}u_i\prod_{j\notin P}v_j,\quad H^+=EH^-,
 \quad M_u=\sum_i X_iu_i,\quad M_v=\sum_i X_iv_i .
\]
The vanishing sums give $LH^-=0$, and the weight $-1$ gives $LH^+=H^-$ and $EH^+=0$. Hence
\begin{equation}
 F_M=\ang{M_u}{H^-}-\ang{M_v}{H^+}.
 \label{eq:seven-v3-explicit-momentum}
\end{equation}
is invariant and belongs to $\mathfrak M_7$. Its $C_a^P$ are $\ang{X_a}{H_P}$, since only the first term contributes to $u_a^2$. Equation~\eqref{eq:seven-v3-reconstruction} now gives $F=F_M$ as polynomials.
\end{proof}
\begin{proposition}\label{prop:seven-v3-shape}\label{subsec:seven-v3-shape}
For $[F]\in W_7$ there are unique constants $k_P$ such that
\begin{equation}
 E_{abc}=\varepsilon(a,b,c;Q)\,k_P\Delta_Q
             \bigl(\ang{a}{b}\ang{b}{c}\ang{c}{a}\bigr)^2.
 \label{eq:seven-v3-shape}
\end{equation}
where $\Delta_Q=\prod_{i<j\in Q}\ang ij$ and $\varepsilon$ is the sign of the concatenation $(a,b,c,Q_{\mathrm{inc}})$.
\end{proposition}
\begin{proof}
For $i,j\in Q$, generically on $\ang ij=0$ any three of the four $Q$ columns span the plane, so their kernel contains a vector with no zero entry. The vanishing of $F$ on this pair and \eqref{eq:seven-v3-geometric-detector} force $\ang ij\mid E_{abc}$. The six factors in $\Delta_Q$ exhaust degree three on each $Q$ leg. The quotient is a three-spinor invariant of degrees $(4,4,4)$, forcing exponent two on each of the three brackets. The sorting signs induce the signed permutation action on the $k_P$.
\end{proof}
\subsubsection{A double collision constrains the triple data}
\label{subsec:seven-v3-section}\label{subsec:seven-v3-straightening}
To recover a polynomial representative from arbitrary $k_P$, start with
\begin{equation}
 C_a^{P,0}=\frac{E_{abc}}{3\ang{b}{c}},\qquad
 C_b^{P,0}=\frac{E_{abc}}{3\ang{c}{a}},\qquad
 C_c^{P,0}=\frac{E_{abc}}{3\ang{a}{b}}.
 \label{eq:seven-v3-C0}
\end{equation}
and let
\begin{equation}
 Q_a^{ap}=\sum_{c\notin\{a,p\}}C_a^{\{a,p,c\},0},\qquad
 R_{ap}=-\frac{Q_a^{ap}X_p-Q_p^{ap}X_a}{\ang{a}{p}}.
 \label{eq:seven-v3-R}
\end{equation}
Both residuals contain $\ang ap^2$, so $R_{ap}=R_{pa}$ is polynomial and $\ang{X_a}{R_{ap}}=-Q_a^{ap}$.
\begin{lemma}\label{lem:seven-v3-zero-total}
$\sum_{a<p}R_{ap}=0$.
\end{lemma}
\begin{proof}
For $k$ supported on one triple, only its three edges contribute. The coefficient of $X_a$ in their sum is $E_{abc}/(3\ang ab\ang ca)+E_{abc}/(3\ang ac\ang ab)=0$. The other two coefficients cancel cyclically. Sum over triples.
\end{proof}
Write $S_a=\sum_{p\ne a}R_{ap}$. Lemma~\ref{lem:triangle-right-inverse} gives $\sum_c H_{apc}=R_{ap}$ for
\begin{equation}
 H_{abc}=\frac{R_{ab}+R_{ac}+R_{bc}}3
              -\frac{S_a+S_b+S_c}{12}
 \label{eq:seven-v3-H}
\end{equation}
Set $C_a^P=C_a^{P,0}+\ang{X_a}{H_P}$. Schouten preserves the detector and $\ang{X_a}{R_{ap}}$ cancels the residual in every Ward identity. With $f_{a0}=\sum_{P\ni a}C_a^P\prod_{i\in P-\{a\}}u_i\prod_{j\notin P}v_j$, put
\begin{equation}\label{eq:seven-v3-section}
 F(k)=\sum_a\left(u_a^2 f_{a0}-u_av_a E_{\ne a}f_{a0}
                      +\tfrac12v_a^2 E_{\ne a}^2f_{a0}\right).
\end{equation}
The formulas use only contractions, unordered sums, and the sorting sign in \eqref{eq:seven-v3-shape}. Consequently $F(\sgn(\sigma)\,\sigma k)=\sigma F(k)$, where $(\sigma k)_P=k_{\sigma^{-1}P}$. By the detector theorem every candidate with these data equals $F(k)$ modulo momentum. Evaluate at the conserving point
\begin{equation}
 \begin{array}{c|rrrrrrr}
 i&1&2&3&4&5&6&7\\\hline
 z_i&0&0&1&1&2&3&4\\
 u_i&1&1&-1&-1&-1&0&1\\
 v_i&0&0&1&1&-4&2&0
 \end{array}
 \qquad X_i=(1,z_i),\quad Y_i=(u_i,v_i).
 \label{eq:seven-v3-double-point}
\end{equation}
\begin{lemma}\label{lem:seven-v3-double-identity}
For $d_{ij}=k_{1ij}-k_{2ij}$ the evaluation is
\begin{equation}
 F(k)(p)=-768\left[(d_{36}-d_{46})+5(d_{37}-d_{47})\right].
 \label{eq:seven-v3-double-identity}
\end{equation}
\end{lemma}
\begin{proof}
The point is fixed by $(12),(34)$; signed equivariance makes the evaluation coefficients antisymmetric separately under each swap. Each nonzero term uses one label from each pair. The three representatives $135,136,137$ give $0,-768,-3840$ by the finite sum in Appendix~\ref{app:seven-v3-double}.
\end{proof}
A genuine candidate vanishes at $p$ and at $(67)p$. Thus $A=d_{36}-d_{46}$, $B=d_{37}-d_{47}$ obey $A+5B=5A+B=0$, forcing $A=B=0$. Relabeling yields
\begin{equation}
 k_{\{a,c,e\}}-k_{\{b,c,e\}}
 -k_{\{a,d,e\}}+k_{\{b,d,e\}}=0
 \quad(a,b,c,d,e\text{ distinct}).
 \label{eq:seven-v3-four-term}
\end{equation}
By Proposition~\ref{thm:subset-interpolation}, $k_{abc}=c_a+c_b+c_c$ with $c_a=s_a/10-s/30$.
\subsubsection{Generators and lifting}
\label{subsec:seven-v3-generators}
Euler's identity $x_a r_{a,x}+y_a r_{a,y}=3r_a$ in the determinant and coprimality of $x_a,y_a$ remove the denominator in \eqref{eq:seven-v3-generators}. In any direction $\eta$ this gives
\begin{equation}
 \det(r_1,\ldots,r_7,D_\eta r_a)=\ang{X_a}{\eta}N_a.
 \label{eq:seven-v3-directional}
\end{equation}
The column representation has determinant one and $\ang{X_a}\eta$ is invariant, giving $G$ invariance and the marked degrees. At any free pair zero two original rows are proportional. Since $P_{ij}$ is prime and $x_a\notin P_{ij}$, $N_a\in\bigcap_{i<j}P_{ij}$.
In the chart $x_i=1$, separate the two colors. For $a\in P$ the negative block contributes $\Delta_Q$ and the positive block is a four-point Vandermonde with the fourth node differentiated at $z_a$. Its determinant is $(z_b-z_a)^2(z_c-z_a)^2(z_c-z_b)$. Including the sorting sign and the factor $\ang bc$ in the detector gives
\begin{equation}
 k_P(N_a)=\mathbf1_{a\in P}.
 \label{eq:seven-v3-normalization}
\end{equation}
For $a\notin P$ there is no surviving marked term. Thus every $F$ has the data of $\sum_a c_a N_a$ and their difference is a momentum relation. The unique coefficients prove independence. The same linear combination lifts to the free pair intersection, yielding the seven-point part of \eqref{eq:main-obstructions} by exactness of invariant projection. Row permutations send $N_a$ to $\sgn(\sigma)N_{\sigma(a)}$ so the unique alternating line is spanned by $\sum_aN_a$.

\subsection{Evaluation at the seven-point double collision}
\label{app:seven-v3-double}
For data supported on a single triple $P$, only its edges contribute and their corrections sum to zero. Thus \eqref{eq:seven-v3-H} specializes to
\begin{equation}
 H_A=
 \begin{cases}
 0,&|A\cap P|=0\text{ or }3,\\[2pt]
 R_{P\setminus\{i\}}/12,&A\cap P=\{i\},\\[2pt]
 R_{ij}/4,&A\cap P=\{i,j\}.
 \end{cases}
 \label{eq:seven-v3-sparse-H}
\end{equation}
At \eqref{eq:seven-v3-double-point}, take the polynomial values of $C_a^{P,0}$ and $R_{ij}$ from \eqref{eq:seven-v3-C0} and \eqref{eq:seven-v3-R}, simplifying before specialization. With $P=13j$ the three nonzero edges are
\[
\begin{array}{c|r|rrr}
P&E_P&R_{13}&R_{1j}&R_{3j}\\\hline
135&-288&(144,96)&(0,96)&(-144,-192)\\
136&1728&(-480,-288)&(96,-288)&(384,576)\\
137&-1728&(336,192)&(-96,192)&(-240,-384)
\end{array}
\]
For instance $\Delta_{2467}=72$ and $E_{135}=-72(1\cdot1\cdot(-2))^2=-288$. Hence $C^{135,0}=(-96,48,-96)$, and $R_{13}=144\,(1,0)+96\,(0,1)$. For any triple $A$ the corrected coefficient is
\begin{equation}
 C_a^A=C_a^{A,0}+(H_A)_2-z_a(H_A)_1.
 \label{eq:seven-v3-affine-C}
\end{equation}
For each mark choose right coordinates $(U,V)$ of determinant one with $U_a\ne0$. Sending $V_i$ to $V_i-(V_a/U_a)U_i$ gives $V_a'=0$ and $V_j'=\sqr aj/U_a$. The evaluation is
\begin{equation}
 F_a=U_a^{-2}\sum_{A\ni a}C_a^A
       \prod_{i\in A\setminus\{a\}}U_i
       \prod_{j\notin A}\sqr{a}{j}.
 \label{eq:seven-v3-bracket-reconstruction}
\end{equation}
Use $U=u$ for $a=1,2,3,4,7$, $U=u+v$ for $a=5$ and $U=v$ for $a=6$. For instance, $a=3$ contributes $10(C_3^{134}+C_3^{234}+C_3^{347})+2C_3^{345}$. For $P=135$ these coefficients are $(-12,8,8,-12)$, giving $F_3=16$. The three rows below follow by the same sum, with the entries divided by $16$:
\[
\begin{array}{c|rrrrrrr|r}
P&F_1/16&F_2/16&F_3/16&F_4/16&F_5/16&F_6/16&F_7/16&F/16\\\hline
135&8&8&1&1&-32&30&-16&0\\
136&-24&-24&-14&-14&0&-12&40&-48\\
137&-24&-24&-27&-27&-96&30&-72&-240
\end{array}
\]
The totals give the coefficients in \eqref{eq:seven-v3-double-identity}.

\subsection{Eight-point completeness in the alternating sector}
\label{sec:eight}\label{sec:v3:eight}
\subsubsection{Fusion and the three divisibility bounds}
\label{subsec:v3:eight-fusion}\label{subsec:v3:eight-filtration}
At $\tau_8=(34,10;3^8)$, the soft representatives of Theorem~\ref{thm:soft-lift} have two extra right degrees. With $V_{a,b}=\Sym^a\C^2\boxtimes\Sym^b\C^2$, put $V=V_{4,1}$, $A=V_{6,3}$, $B=V_{5,2}$ and $T=V_{1,1}$. The extras lie on one leg or two, giving the alternating source
\begin{equation}
 \mathcal A_8=
 [\bigwedge^7V\otimes A]^G
 \oplus[\bigwedge^6V\otimes\bigwedge^2B]^G.
 \label{eq:v3:eight-source}
\end{equation}
On a dense pair locus, little-group scaling and homogeneity reduce the two marked columns to
\begin{equation}
 X_1=X_2=X_0,\qquad
 Y_1=\frac{Y_0}{1+t},\qquad Y_2=\frac{tY_0}{1+t}.
 \label{eq:v3:eight-fusion}
\end{equation}
The momentum equation becomes the seven-leg equation. Terms with equal colliding degrees vanish by alternation. The remaining two positions of the mark give $(t-t^3)/(1+t)^4$ or $(t-t^2)/(1+t)^3$, the same scalar. Thus
\begin{equation}
 F\big|_{\eqref{eq:v3:eight-fusion}}
 =\frac{t(1-t)}{(1+t)^3}\,\beta(F),
 \qquad\beta:\mathcal A_8\longrightarrow\mathcal D_8.
 \label{eq:v3:eight-beta}
\end{equation}
where
\begin{equation}
 \mathcal D_8=
 [V_{10,4}\otimes\bigwedge^6V]^G
 \oplus[V_{9,3}\otimes B\otimes\bigwedge^5V]^G.
 \label{eq:v3:eight-target}
\end{equation}
In particular $\beta(F)$ belongs to $K_\partial=L_0^3\prod_{i=3}^8L_i$. Density and the prime free pair ideal give
\begin{equation}
 \mathcal Z_8=\mathcal A_8\cap\bigcap_{i<j}P_{ij}.
 \label{eq:v3:eight-free-kernel}
\end{equation}
Here $\mathcal Z_8=\ker\beta$. In the double-marked summand fix the marked pair $a,b$ and write $O$ for the six other legs. A mixed fusion leaves the remaining marked label at $Y$ degree two; its multidegree distinguishes it from every other such label. Thus the fixed-pair component must satisfy each mixed free pair condition separately. Equal-type pairs satisfy theirs by alternation and homogeneity. Denote this fixed-pair kernel by $Z$.
\begin{lemma}\label{lem:v3:eight-layers}
For $Z_j=Z\cap(\sqr ab^j)$, $\dim Z_0/Z_1\le1$, $\dim Z_1/Z_2\le2$ and $Z_2=0$.
\end{lemma}
\begin{proof}
(i) Set $Y_a=Y_b=e_+=(1,0)$ and read ordinary multilinear coefficients at $e_+,e_-=(0,1)$. Zero right weight forces exactly one ordinary positive leg $c$ and five negative legs $T'$. The ten same-color brackets among $T'$ divide the restriction, exhausting angle degree four on each. The three residual degrees $(5,5,4)$ force
\begin{equation}\label{eq:v3:eight-layer0}
 k_c\Delta_{T'}\ang ab^3\ang ac^2\ang bc^2,\qquad
 \Delta_{T'}=\prod_{i<j\in T'}\ang ij .
\end{equation}
Ordinary alternation relates all $k_c$. If they vanish, covariance and homogeneity give vanishing on the full proportional-$Y_a,Y_b$ locus, hence divisibility by $\sqr ab$.
(ii) For $F=\sqr ab G_1$, the marked legs of $G_1$ are symmetric, each linear in $Y$. The original $P_{ab}$ condition is already supplied by $\sqr ab$, so $G_1$ need not contain an extra $\ang ab$ factor. Primality cancels the square-bracket factor from every other pair condition. On $Y_a=Y_b=e_+$ there are now two further positive legs $c,d$ and four negative legs $T'$. The forced angle factor is
\begin{equation}\label{eq:v3:eight-layer1-factor}
 \Delta_{T'}\ang ac\ang ad\ang bc\ang bd\ang cd .
\end{equation}
Its quotient belongs to $[\Sym^2(V_3)\otimes V_2\otimes V_4]^{SL_2}$, $V_j=\Sym^j\C^2$: the residual degrees are three at $a,b$ and one elsewhere, with symmetry on $\{a,b\},\{c,d\},T'$. Since $\Sym^2V_3=V_6\oplus V_2$ and $V_2\otimes V_4=V_6\oplus V_4\oplus V_2$, the bound is two. The common kernel is divisibility of $G_1$ by $\sqr ab$ as in (i). Alternation of the ordinary legs fixes the restriction for every other choice of the two positive labels $c,d$.
(iii) For $F=\sqr ab^2 G_2$, $G_2$ is independent of $Y_a,Y_b$. Canceling the square factor in $P_{ac}$ and setting $Y_a=Y_c$ shows $\ang ac\mid G_2$ for every $c\in O$. Six factors exceed degree five in $X_a$, so $G_2=0$.
\end{proof}
\subsubsection{Counting source and boundary relations}
\label{subsec:v3:eight-counts}\label{subsec:v3:eight-upper}
Write $\mathcal J_8=\mathcal A_8\cap\mathfrak M_8$ and $\mathcal H_8=\mathcal D_8\cap\mathfrak M_7$. Define
\[
 M_k(p,q)=[s^kx^py^q](1-x^{-2})(1-y^{-2})
 \prod_{i=0}^{4}\prod_{j=0}^{1}(1+s x^{4-2i}y^{1-2j}) .
\]
For $p,q\ge0$ the Weyl weight differences give the multiplicity of $V_{p,q}$ in $\bigwedge^k V$. Expanding the product gives
\[
\begin{array}{c|l|l}
k&(p,q)&M_k(p,q)\ \text{in order}\\\hline
3&(6,3),(6,1),(4,3),(4,1)&1,1,0,2\\
2&(2,0),(0,2)&0,0\\
4&(10,4),(10,2),(8,4),(8,2)&0,1,0,1
\end{array}
\]
Here $\bigwedge^{10-k}V\simeq(\bigwedge^k V)^*$. The single-marked summand has dimension $M_3(6,3)=1$, so projection onto it gives $\dim\mathcal Z_8\le1+\dim Z\le4$. The source first and second momentum coefficient spaces are $[\bigwedge^7 V\otimes B\otimes T]^G$ and $[\bigwedge^8 V\otimes\bigwedge^2 T]^G$. By Clebsch--Gordan, $B\otimes T$ has the four types in row one and $\bigwedge^2T$ the two types in row two. Weighted exactness (Theorem~\ref{thm:weighted-koszul}) therefore gives $\dim\mathcal J_8=4$. On the fused side the minimum soft degree is nine, leaving only $[V_{9,3}\otimes\bigwedge^6 V\otimes T]^G$ at first order and zero at second order. The last row gives $\dim\mathcal H_8=2$.
Substitution in $F\in P_{12}+\mathfrak M_8$ gives $\beta(F)\in\mathfrak M_7$ after canceling the common nonzero rational factor in $t$. Since the image and ideal are defined over $\C$, the inclusion descends from $\C(t)$. Thus $W_8^{\sgn}\subseteq\beta^{-1}(\mathcal H_8)/\mathcal J_8$ and Proposition~\ref{thm:boundary-exact} gives
\begin{align}
 \dim W_8^{\mathrm{sgn}}
 &\le\dim\ker\beta+
      \dim(\operatorname{im}\beta\cap\mathcal H_8)-\dim\mathcal J_8\notag\\
 &\le4+2-4=2.
 \label{eq:v3:eight-upper}
\end{align}
\subsubsection{Polynomial generators and their independence}
\label{subsec:v3:eight-generators}\label{subsec:v3:eight-curve}\label{subsec:v3:eight-completeness}
Use $r_i=(x_i^{4-j}y_i^j u_i)_{j=0}^4\oplus(x_i^{4-j}y_i^j v_i)_{j=0}^4$ in this column order, and its unit volume form $\operatorname{vol}$. Define $\mathcal R=\sum_aR_a$, $\mathcal S=\sum_{a<b}S_{ab}$, where
\begin{equation}
 \begin{aligned}
 R_a&=\frac{\det(r_1,\ldots,r_8,\partial_{y_a}r_a,\partial_{y_a}^2r_a)}{x_a^3},\\[3pt]
 B_a&=\frac{r_a\wedge\partial_{y_a}r_a}{x_a}.
 \end{aligned}
 \label{eq:v3:eight-RB}
\end{equation}
\begin{equation}
 S_{ab}=(-1)^{a+b}\operatorname{vol}\left(
 \bigwedge_{i\notin\{a,b\}}^{\mathrm{increasing}}r_i\wedge
 \left(\partial_{x_a}B_a\wedge\partial_{y_b}B_b
       -\partial_{y_a}B_a\wedge\partial_{x_b}B_b\right)\right),
 \quad a<b.
 \label{eq:v3:eight-Sab}
\end{equation}
\begin{lemma}\label{lem:v3:eight-generators}
Both sums are invariant, alternating polynomials of degree $\tau_8$ in the free pair intersection.
\end{lemma}
\begin{proof}
A nonzero minor of $(r,r_y)$ uses distinct powers $p,q$ of $y$, so its $x$ exponent is $8-p-q\ge1$. For $(r,r_y,r_{yy})$ three distinct powers give $12-p-q-s\ge3$. Thus the displayed quotients are polynomial. In a direction $\eta$, Euler's identity gives $r\wedge D_\eta r=\ang X\eta B$ and $r\wedge D_\eta r\wedge D_\eta^2r=\ang X\eta^3(r\wedge r_y\wedge r_{yy})/x^3$. The blocks are covariant, the derivative contraction is invariant, and the column representation has determinant one. The single and double marks have degrees $(6,3)$ and $(5,2)$, respectively. Ordinary row exchange gives a minus sign; swapping two marked two-form blocks has exterior sign $+1$ and derivative sign $-1$. The prefactor $(-1)^{a+b}$ accounts for moving the blocks past ordinary rows. Both sums are alternating. Finally, the identities
\begin{equation}
 \partial_{y_a}B_a=\frac{r_a\wedge\partial_{y_a}^2r_a}{x_a},
 \qquad
 \partial_{x_a}B_a=
 \frac{r_a\wedge(6\partial_{y_a}r_a-y_a\partial_{y_a}^2r_a)}{x_a^2}.
 \label{eq:v3:eight-Bderivatives}
\end{equation}
show that each term retains all eight original rows. Two rows become proportional on a free pair locus; primality (Lemma~\ref{lem:free-pair-prime}) extends this vanishing globally.
\end{proof}
Use the cubic slice
\begin{equation}
 X_i=(1,z_i),\qquad Y_i=c_i(1,f(z_i)),\qquad
 f(z)=z^3-\sigma z.
 \label{eq:v3:eight-cubic}
\end{equation}
Set $m_k=\sum_i c_i z_i^k$, $P_c=\prod_i c_i$, $D_z=\prod_{i<j}(z_j-z_i)$, $p_i=3z_i^2-\sigma$, $q_i=6z_i$. Momentum requires
\begin{equation}
 m_0=m_1=m_3=0,\qquad m_4=\sigma m_2.
 \label{eq:v3:eight-moments}
\end{equation}
On $x=1$ write $Y=c(1,w)$. Replace the five $w$-columns by $wz^2,wz^3,wz^4,w-f,z(w-f)$, a determinant-one operation. On $w=f$ the first eight columns reduce triangularly to powers $0,\ldots,7$, contributing $P_cD_z$. Holding $w$ fixed before restriction, the derivatives of the two relation columns are $\mathbf f_i=(-p_i,-z_ip_i)$ and $\mathbf g_i=(-q_i,-2p_i-z_iq_i)$. Thus the residual determinant for $R_a$ gives
\begin{equation}
 \frac{R_a}{P_cD_z}=2c_a^2p_a^2.
 \label{eq:v3:eight-Rcurve}
\end{equation}
With $\delta=z_b-z_a$, formula~\eqref{eq:v3:eight-Bderivatives} gives
\begin{align}
 \frac{S_{ab}}{P_cD_z}
 &=c_ac_b\left(6\det(\mathbf f_a,\mathbf g_b)
 -6\det(\mathbf g_a,\mathbf f_b)
 +\delta\det(\mathbf g_a,\mathbf g_b)\right)\notag\\
 &=c_ac_b\left[24p_ap_b+4\delta(p_aq_b-p_bq_a)+\delta^2q_aq_b\right].
 \label{eq:v3:eight-Scurve}
\end{align}
Summing uses $4\delta(p_aq_b-p_bq_a)+\delta^2q_aq_b=-12(3z_az_b+2\sigma)\delta^2$. Expanding the pair sums as moments yields
\begin{align}
 \frac{\mathcal S+6\mathcal R}{P_cD_z}
 &=12(3m_2-\sigma m_0)^2
 -36(m_1m_3-m_2^2)-24\sigma(m_0m_2-m_1^2)\notag\\
 &=144m_2^2\quad\text{on the momentum variety}.
 \label{eq:v3:eight-square}
\end{align}
Choose eight distinct nonzero real nodes symmetric about zero, $h(z)=\prod_i(z-z_i)$, and $\sigma=\tfrac12\sum z_i^2$. The expansion of $1/h(z)=\sum_i[h'(z_i)(z-z_i)]^{-1}$ at infinity gives
\begin{equation}
 \sum_i\frac{z_i^k}{h'(z_i)}=
 \begin{cases}
 0,&0\le k\le6,\\1,&k=7,\\0,&k=8,\\\sigma,&k=9.
 \end{cases}
 \label{eq:v3:eight-lagrange}
\end{equation}
Both $c_i=1/h'(z_i)$ and $c_i=z_i^5/h'(z_i)$ obey the momentum equations. On the first family, $m_2=0$ and $\mathcal R/(P_cD_z)=2\sum_i c_i^2 p_i^2>0$. On the second, $m_2=1$. Consequently $\mathcal R$ and $\mathcal S+6\mathcal R$ are independent in the momentum quotient.
\begin{theorem}\label{thm:v3:eight-complete}
The classes $[\mathcal R],[\mathcal S]$ form a basis of $W_8^{\sgn}$. Every such class lifts to the free pair intersection.
\end{theorem}
\begin{proof}
The two independent free generators attain \eqref{eq:v3:eight-upper}, proving the lifting assertion including the obstruction component in \eqref{eq:main-obstructions}. Let $\mathcal J_{8,0}=\mathcal J_8\cap\ker\beta$. Its dimension is at least $4-2=2$, whereas the two generators remain independent modulo it. Together with $\dim\mathcal Z_8\le4$, this gives
\[
 2\le\dim(\mathcal Z_8/\mathcal J_{8,0})
 =\dim\mathcal Z_8-\dim\mathcal J_{8,0}\le4-2=2.
\]
Equality throughout yields
\begin{equation}
 \begin{gathered}
 0\longrightarrow\mathcal J_{8,0}\longrightarrow\mathcal J_8
   \xrightarrow{\beta}\mathcal H_8\longrightarrow0,\\
 0\longrightarrow\mathcal J_{8,0}\longrightarrow\mathcal Z_8
   \longrightarrow W_8^{\mathrm{sgn}}\longrightarrow0,\\
 (\dim\mathcal J_{8,0},\dim\mathcal J_8,\dim\mathcal H_8,\dim\mathcal Z_8)
 =(2,4,2,4).
 \end{gathered}
 \label{eq:v3:eight-two-exact}
\end{equation}
\end{proof}
The seven-point and alternating eight-point theorems also give the lifting identities
\begin{equation}
 \left[\frac{\bigcap_{i<j}(P_{ij}+\mathfrak M_7)}{I_7+\mathfrak M_7}\right]^G_{\tau_7}=0,
 \qquad
 \left[\frac{\bigcap_{i<j}(P_{ij}+\mathfrak M_8)}{I_8+\mathfrak M_8}\right]^{G,\sgn}_{\tau_8}=0.
 \label{eq:main-obstructions}
\end{equation}
\subsubsection{Hodges coordinates on the cubic slice}
\label{subsec:v3:eight-hodges}
On the same slice write $G_c=\operatorname{diag}(c_i)$, $U_i=(1,z_i,z_i^2)$, and
\[
 A_0=\begin{pmatrix}-\sigma&0&1\\0&1&0\\1&0&0\end{pmatrix},\qquad
 \Phi=-m_2 G_c+G_c U A_0 U^{\mathsf t}G_c,\qquad
 U^{\mathsf t}G_cU=m_2 A_0^{-1}.
\]
The last two identities follow by taking the reference $(0,1)$ and using \eqref{eq:v3:eight-moments}. Delete a triple $D$ and retain $I=D^c$. The determinant lemma yields
\[
 \det\Phi_I=(-m_2)^5\prod_{i\in I}c_i\,
 \det(m_2^{-1}A_0 U_D^{\mathsf t}G_{c,D}U_D)
 =P_c m_2^2(\det U_D)^2 .
\]
The identity extends to $m_2=0$ by polynomial continuation on the linear momentum slice. Dividing by the squared deleted Vandermonde gives $N_8^{\mathrm H}/(P_cD_z)=m_2^2$. By Proposition~\ref{prop:hodges-member}, $N_8^{\mathrm H}$ belongs to the classified space. The two independent evaluations above then fix its global coordinate:
\begin{equation}
 \boxed{N_8^{\mathrm H}=\frac{6\mathcal R+\mathcal S}{144}.}
 \label{eq:v3:eight-Hodges-coordinate}
\end{equation}

\section{Hodges determinants and physical uniqueness}
\label{app:physical}
\subsection{Hodges minors, polynomiality and the collinear residue}
\label{app:hodges-tech}
\subsubsection{Three null directions and a reduced determinant}
Choose a reference spinor \(q\), initially with \(\ang{i}{q}\ne0\), and define
\begin{equation}
 \Phi_{ij}=\frac{\sqr{i}{j}}{\ang{i}{j}}\quad(i\ne j),\qquad
 \Phi_{ii}=-\sum_{j\ne i}
 \frac{\sqr{i}{j}\ang{j}{q}^{2}}{\ang{i}{j}\ang{i}{q}^{2}}.
 \label{eq:Phi-definition}
\end{equation}
Momentum conservation implies that the diagonal is independent of \(q\). For another spinor \(r\), the difference of the two squared ratios has a factor \(\ang{i}{j}\) by Schouten; after cancellation the sum is a linear combination of \(\sum_j\sqr{i}{j}\ang{j}{q}\) and \(\sum_j\sqr{i}{j}\ang{j}{r}\), both zero.

The three vectors
\begin{equation}
 v_i^{(0)}=\ang{i}{q}^2,\qquad
 v_i^{(1)}=\ang{i}{q}\ang{i}{r},\qquad
 v_i^{(2)}=\ang{i}{r}^2
 \label{eq:Phi-null}
\end{equation}
lie in the kernel of \(\Phi\). The first is immediate from the diagonal definition. For the second use
\[
 \ang{j}{r}\ang{i}{q}-\ang{j}{q}\ang{i}{r}
 =\ang{i}{j}\ang{q}{r};
\]
the remaining sum is \(\ang{q}{r}\sum_j\sqr{i}{j}\ang{j}{q}/\ang{i}{q}=0\). Factoring a difference of squares proves the third in the same way. These vectors are evaluations of binary quadratics at the external angle spinors.

For increasing triples \(R,C\), let \(\Phi^R_C\) be the matrix after deleting those rows and columns, and set
\begin{equation}
 \Delta_{abc}=\ang{a}{b}\ang{b}{c}\ang{c}{a},\qquad
 \det{}'\Phi=
 \frac{(-1)^{\sum R+\sum C}\det\Phi^R_C}{\Delta_R\Delta_C}.
 \label{eq:reduced-determinant}
\end{equation}
This expression is independent of \(R,C\). Indeed the exterior cofactor matrix of an operator with a three-dimensional specified kernel factors through the top exterior power of that kernel. Its entries are proportional to the products of the corresponding three-row kernel minors. The three-row evaluation determinant of \((x_i^2,x_iy_i,y_i^2)\) is \(-\Delta_{abc}\); the two signs cancel in the product. If the corank is larger, all the displayed minors vanish, so the same identity holds. This argument applies on the generic rank-two angle chart and hence as a rational identity in the domain \(Q_n\).

Define \(N_n^{\mathrm H}=D_n\det{}'\Phi\).

\begin{lemma}[Cancellation of all excessive angle poles]
\label{lem:hodges-polynomial}
For \(n\ge5\), \(N_n^{\mathrm H}\) is a polynomial class in \(Q_n\), with degree \(\tau_n\), and is alternating in the labels.
\end{lemma}
\begin{proof}
First choose a deleted label as the reference \(q\). The remaining entries and the deleted-triple denominators then involve only angle brackets, so \(\det'\Phi\) belongs to the localization of \(Q_n\) obtained by inverting their product.

Fix a pair \(a,b\). Choose two other labels \(c,d\) and use \(R=(a,c,d)\), \(C=(b,c,d)\), with signs adjusted to increasing order and reference \(q=X_c\). At a generic point of \(\ang{a}{b}=0\), both deleted-triple factors and all reference denominators are nonzero. Row \(a\) and column \(b\) are absent, so the diagonal entries involving the pair \(a,b\) are absent. Only the entry in row \(b\), column \(a\), can contain \(1/\ang{a}{b}\). A determinant uses that entry at most once. Thus \(\det'\Phi\) has at most a simple pole along this divisor.

By Proposition~\ref{prop:angle-prime}, every angle bracket is a prime element, and the different brackets are nonassociate. Write one common expression as \(A/\prod_{i<j}\ang{i}{j}^{m_{ij}}\). The simple-pole statement implies that \(A\) is divisible by \(\ang{i}{j}^{m_{ij}-1}\) whenever \(m_{ij}>1\). Corollary~\ref{cor:angle-denominator-clearing} justifies contracting each cancellation from the prime localization back to \(Q_n\); the factors may then be removed successively. Multiplication by \(D_n\) therefore leaves an element of \(Q_n\).

Each term of the reduced minor has square degree \(n-3\), angle degree \(-(n+3)\), and weight \(-4\) on each leg. This gives \eqref{eq:target-degree}. Simultaneous label permutations leave \(\det'\Phi\) unchanged by deletion independence, while \(D_n\) changes by their sign.
\end{proof}

\subsubsection{The collinear residue from a rank-one singular block}
Fix \(a=n-1,b=n\) and a nonzero ratio \(c\). On the marked holomorphic boundary define the fused spinors by
\begin{equation}
 X_a=cX_b,\qquad X_P=X_b,\qquad Y_P=cY_a+Y_b.
 \label{eq:marked-fusion}
\end{equation}
They obey \(p_P=p_a+p_b\). The notation \(\iota_{ab,c}\) will mean substitution of these fused spinors in an \((n-1)\)-point function.

\begin{proposition}[Determinant splitting identity]
\label{prop:hodges-splitting}
On the generic marked boundary, for \(n\ge5\),
\begin{equation}
 \lim_{\ang{a}{b}\to0}\ang{a}{b}\det{}'\Phi_n
 =-c^{-2}\sqr{a}{b}\,\iota_{ab,c}(\det{}'\Phi_{n-1}).
 \label{eq:det-splitting}
\end{equation}
Consequently
\begin{equation}
 \rho_{ab,c}^{\rm gen}(N_n^{\mathrm H})
 =-c^{n-4}\sqr{a}{b}
 \left(\prod_{i=1}^{n-2}\ang{i}{b}\right)
 \iota_{ab,c}(N_{n-1}^{\mathrm H}).
 \label{eq:numerator-splitting}
\end{equation}
For \(n=5\), the last expression is interpreted through the rational four-point reduced determinant; all generic factors are well-defined.
\end{proposition}
\begin{proof}
Delete the same three labels, all outside \(a,b\), from rows and columns. Let \(\epsilon=\ang{a}{b}\) and \(t=\sqr{a}{b}/\epsilon\). The singular part of the retained \(a,b\) block is
\[
 t\begin{pmatrix}-c^{-2}&1\\1&-c^2\end{pmatrix}
 =-t\begin{pmatrix}c^{-1}\\-c\end{pmatrix}
       \begin{pmatrix}c^{-1}&-c\end{pmatrix}.
\]
For the calculation off the boundary, replace \(c\) here by
\(c_q=\ang{a}{q}/\ang{b}{q}\), so the singular block is exactly rank one. Add \(c_q^2\) times row \(a\) to row \(b\), and do the same column operation. Both operations preserve the determinant. They remove the singular entries from row and column \(b\); the only remaining singular entry is \(-t/c_q^2\) in position \(a,a\).

In the boundary limit, the entries linking the transformed row \(b\) to an ordinary leg \(k\) are
\[
 c^2\frac{\sqr{a}{k}}{c\ang{b}{k}}
 +\frac{\sqr{b}{k}}{\ang{b}{k}}
 =\frac{\sqr{P}{k}}{\ang{P}{k}}.
\]
The transformed diagonal has ordinary-leg contribution
\[
 -\sum_{k\ne a,b}
 \frac{(c\sqr{a}{k}+\sqr{b}{k})\ang{k}{q}^2}
 {\ang{b}{k}\ang{b}{q}^2},
\]
which is the fused diagonal. Each other diagonal combines the \(a\)- and \(b\)-terms into the same fused contribution. Hence the coefficient of \(t\) in the determinant is \(-c^{-2}\) times the lower-point retained determinant. The deleted-triple factors are unchanged, proving \eqref{eq:det-splitting}.

Finally, cancelling the one vanishing factor before restriction gives
\begin{align*}
 \left.\frac{D_n}{\ang{a}{b}}\right|_{X_a=cX_b}
 &=c^{n-2}D_{n-2}\prod_{i=1}^{n-2}\ang{i}{b}^{2}\\
 &=c^{n-2}\iota_{ab,c}(D_{n-1})\prod_{i=1}^{n-2}\ang{i}{b}.
\end{align*}
Multiplying \eqref{eq:det-splitting} by this expression proves \eqref{eq:numerator-splitting}.
\end{proof}

\begin{proposition}[Global Hodges membership]
\label{prop:hodges-member}
For every \(n\ge5\), \(N_n^{\mathrm H}\in W_n^{\sgn}\).
\end{proposition}
\begin{proof}
Polynomiality, degree, and alternating symmetry were proved in Lemma~\ref{lem:hodges-polynomial}. On a generic pair boundary, \eqref{eq:numerator-splitting} has a factor \(\sqr{a}{b}\), and hence vanishes when this bracket is also zero. Proposition~\ref{prop:pair-prime} says that \(Q_n/J_{ab}\) is a domain. Allowing \(c\) to range over \(\C^*\), the generic marked pair loci used above have dense union in this pair variety, so polynomial vanishing there implies membership in \(J_{ab}\). The fused configuration remains generic even at \(n=5\): split a generic four-point leg as \(X_a=cX_P\), \(Y_a=dY_P\), \(Y_b=(1-cd)Y_P\), with \(d(1-cd)\ne0\). Thus no lower angle denominator is forced to vanish when \(\sqr{a}{b}=0\). Relabelling covers every pair. Primality extends this membership to the zero-spinor and multiple-collision strata.
\end{proof}

\subsubsection{Fixing the seven-point determinant convention}
At seven points label permutations act on \(N_a\) by the sign times the permutation of \(a\). Thus the alternating part is the line \(\C\sum_aN_a\). To fix its scalar, take
\[
 z=(0,1,2,3,4,5,6),\quad X_i=(1,z_i),\quad
 u=(1,-2,1,0,0,0,0),\quad v=(0,0,0,1,-1,-1,1).
\]
Momentum conservation follows from \(\sum u_i=\sum z_i u_i=\sum v_i=\sum z_i v_i=0\). The confluent Vandermonde formula gives
\[
 N_1=N_2=N_3=-96,\qquad N_4=N_5=N_6=N_7=0.
\]
For \(q=(0,1)\), deletion of rows and columns \(1,2,3\) leaves the diagonal matrix
\[
 \operatorname{diag}(-1/3,1/12,1/30,-1/60).
\]
Since \(D_7=24883200\) and \(\Delta_{123}^2=4\), this gives \(N_7^{\mathrm H}=96\). One-dimensionality then gives the global identity
\[
 N_7^{\mathrm H}=-\tfrac13\sum_aN_a.
\]

\subsubsection{Hodges confluent minors for the counterexamples}
\label{subsec:hodges-confluent}
We express the extra numerators through small determinants of the off-diagonal Hodges entries. Confluence duplicates a right spinor while bringing the copied angle spinors to the original one; the copied legs are auxiliary variables in the free spinor ring.

\paragraph{Two equally sized groups.}
Start with $2m$ independent spinors and rows
\[
 r_i=(u_ix_i^{m-1},u_ix_i^{m-2}y_i,\ldots,u_iy_i^{m-1},
      v_ix_i^{m-1},\ldots,v_iy_i^{m-1}).
\]
Let $A=(a_1,\ldots,a_m)$ and $B=(b_1,\ldots,b_m)$ be ordered groups. Set $C_{ba}=[ba]/\ang{b}{a}$, with rows in $B$ and columns in $A$. Then
\begin{equation}
 \det(r_A,r_B)=(-1)^{m(m-1)/2}
 \left(\prod_{a\in A,b\in B}\ang{a}{b}\right)\det C.
 \label{eq:segre-hodges-identity}
\end{equation}
The product clears every denominator on the right, so this is also a polynomial identity.

For a direct proof, work on the dense chart $X_i=(1,z_i)$, $u_i\ne0$, and put $w_i=v_i/u_i$. Factoring $u_i$ from each row gives $[V,\operatorname{diag}(w)V]$, where $V_i=(1,z_i,\ldots,z_i^{m-1})$. With $L=V_BV_A^{-1}$, eliminating the lower left block gives
\[
 \Bigl(\prod_i u_i\Bigr)(\det V_A)^2
 \det\bigl((w_b-w_a)L_{ba}\bigr).
\]
For $h_A(z)=\prod_{a\in A}(z-z_a)$, Lagrange interpolation reads
\[
 L_{ba}=\frac{h_A(z_b)}{(z_b-z_a)h_A'(z_a)},\qquad
 (w_b-w_a)L_{ba}=\frac{h_A(z_b)}{h_A'(z_a)u_bu_a}C_{ba}.
\]
Extracting row and column factors and using
\[
 \prod_{a\in A}h_A'(z_a)=(-1)^{m(m-1)/2}(\det V_A)^2
\]
proves~\eqref{eq:segre-hodges-identity}. Homogeneity in every angle spinor and denominator clearing extend the identity to all spinors.

\begin{proposition}[Hodges confluent minors]\label{prop:hodges-confluent}
Let $n=7$ or $8$, and set $s=n-6$. Choose distinct labels $a,b,c$, let $B$ be their complement in increasing order, and put $\varepsilon_{abc}=\sgn(a,b,c,B)$. For an auxiliary spinor $q$ with $\ang{a}{q}\ne0$, define columns by differentiation of $X_a$ alone:
\[
 C_a^{(k)}(i)=D_q^k\Phi_{ia}
 =(-1)^k k!\frac{[ia]\ang{i}{q}^{k}}{\ang{i}{a}^{k+1}},
 \qquad i\in B,
\]
keeping $Y_a$ and every other spinor fixed. Then
\begin{equation}
 \mathcal G_a=
 \varepsilon_{abc}\,
 \frac{\displaystyle\prod_{i\in B}\ang{a}{i}^{s+1}\ang{b}{i}\ang{c}{i}}
      {\ang{a}{q}^{s(s+1)/2}}
 \det_{i\in B}\left(C_a^{(0)},\ldots,C_a^{(s)},\Phi_{ib},\Phi_{ic}\right)
 \label{eq:hodges-confluent}
\end{equation}
 equals $N_a$ at seven points and $R_a$ at eight points. In particular it is independent of the auxiliary labels $b,c$ and direction $q$, and has a globally polynomial extension.
\end{proposition}
\begin{proof}
In~\eqref{eq:segre-hodges-identity}, set $m=n-3$ and take
\[
 A=(a,a^{(1)},\ldots,a^{(s)},b,c),\qquad
 X_{a^{(j)}}=X_a+t_jq,\quad Y_{a^{(j)}}=Y_a.
\]
Differentiate $j$ times with respect to $t_j$, then set every $t_j=0$. The large determinant becomes the original $n$ rows followed by $D_qr_a$ and, at eight points, $D_q^2r_a$. On the smaller side, any derivative landing on the external angle product leaves repeated derivative orders among the confluent columns, so that term vanishes. Only the term placing every derivative on its corresponding column survives. Here $m=4,5$ gives $(-1)^{m(m-1)/2}=1$. Moving the copied rows to the end contributes $(-1)^6$ or $(-1)^{14}$, respectively, both positive; the remaining sorting sign is $\varepsilon_{abc}$.

Finally, Euler's identity yields
\[
 \det(r_1,\ldots,r_7,D_qr_a)=\ang{a}{q}N_a,
\]
\[
 \det(r_1,\ldots,r_8,D_qr_a,D_q^2r_a)=\ang{a}{q}^3R_a.
\]
Dividing by these factors proves the claim. Appendix~\ref{app:classification} establishes polynomiality and membership in the free pair intersection. The same polynomials therefore define the expressions where the auxiliary denominators vanish.
\end{proof}

Multiplying the two singular diagonal entries on a collinear boundary by $\eta$ gives
\[
 t\begin{pmatrix}-\eta c^{-2}&1\\1&-\eta c^2\end{pmatrix},
 \qquad\det=(\eta^2-1)t^2,\qquad t=[ab]/\ang{a}{b}.
\]
Generic $\eta$ introduces a double angle pole, while $D_n$ cancels only one power.

\subsection{Interpolation proof of the large-shift criterion}
\label{sec:bcfw}

Fix an ordered pair of legs and make the momentum-preserving deformation of Britto, Cachazo, Feng and Witten~\cite{BCFW},
\begin{equation}
 Y_a\longmapsto Y_a+ZY_b,\qquad X_b\longmapsto X_b-ZX_a .
 \label{eq:bcfw-shift}
\end{equation}
For the two negative-helicity legs the stripped factor $\ang{a}{b}^{8}$ is constant under this shift. Of the factors of $D_n$, only the $n-2$ brackets joining $b$ to a spectator depend on $Z$. Thus on generic kinematics ordinary decay of $N_n/D_n$ is the linear constraint $\deg_Z N_n(Z)\le n-3$ on the shifted polynomial.

\subsubsection{The collision probed by the shift}
Independent little-group rescalings by $t_i$ replace $Z$ in the original shift by $Zt_a/t_b$, preserving its order at infinity. We may therefore write $X_i=(1,z_i)^{\mathsf T}$ with distinct $z_i$. Put $d=n-3$, $h(u)=\prod_i(u-z_i)$, and form the two-component polynomial
\begin{equation}
 T(u)=\sum_iY_i\prod_{j\ne i}(u-z_j),\qquad
 Y_i=\frac{T(z_i)}{h'(z_i)} .
 \label{eq:bcfw-map}
\end{equation}
Momentum conservation reads $\sum_iY_i=\sum_i z_iY_i=0$ and kills the coefficients of $u^{n-1}$ and $u^{n-2}$ in $T$. Conversely, Lagrange interpolation assigns a unique conserving set of $Y_i$ to every pair $T=(T_1,T_2)$ with $\deg T_\alpha\le d$. Thus the coefficients of $T$ vary freely at fixed distinct nodes. With $\Delta(z)=\prod_{i<j}(z_j-z_i)$ define
\begin{equation}
 P(z,T)=\Delta(z) N_n\bigl((1,z_i)^{\mathsf T},T(z_i)/h'(z_i)\bigr)
       =\Delta(z)^2\mathcal A_n^{\mathrm{aff}}(z,T).
 \label{eq:bcfw-pullback}
\end{equation}
This expression has total degree $2d$ in the coefficients of $T$. Its denominators are products of node differences. At fixed $T$ a permutation of the nodes permutes the whole external legs, so $P$ is symmetric in the nodes.

After~\eqref{eq:bcfw-shift} the first component of $X_b$ is $1-Z$. Restore it to one by rescaling leg $b$ by $\epsilon=(1-Z)^{-1}$. Its new node and right spinor, and the change of amplitude, are
\begin{equation}
 z_b'=z_a+\epsilon(z_b-z_a),\qquad
 Y_b'=\epsilon^{-1}Y_b,\qquad
 \mathcal A_n^{\mathrm{aff}}(z',T_\epsilon)=\epsilon^{-4}\mathcal A_n(Z).
 \label{eq:bcfw-collision}
\end{equation}
The power $-4$ is the little-group weight of the stripped amplitude. As $Z\to\infty$, the two nodes collide linearly and $\Delta(z')$ has a simple zero. The assumed falloff therefore gives $P(z',T_\epsilon)=O(\epsilon^{-1})$.

There are exactly $n-2=d+1$ spectator nodes. Their spinors stay fixed and give the interpolation data
\[
 T_\epsilon(z_k)=h_\epsilon'(z_k)Y_k,\qquad k\ne a,b,\qquad
 h_\epsilon(u)=\prod_i(u-z_i').
\]
The evaluation matrix on these nodes is invertible and each $h_0'(z_k)$ is nonzero. It follows that $T_\epsilon$ has a finite limit $T_0$. The spectator spinors are arbitrary: since $z_a\ne z_b$, the two momentum equations determine $Y_a,Y_b$ for every choice of the others. Consequently $T_0$ ranges over the full coefficient space for generic colliding nodes. A pole of order at least two would have a nonzero leading coefficient for generic such data, contradicting the bound.

At fixed $T$ write $m=(z_a+z_b)/2$ and $s=z_b-z_a$. Exchanging the two labels fixes $m$ and sends $s$ to $-s$. The Laurent expansion of $P$ at a generic point of the diagonal has only even powers of $s$. The bound of at most a simple pole now gives regularity. Full permutation symmetry carries the result to every pair. This parity argument is the local form, in the variables~\eqref{eq:bcfw-map}, of the improvement of gravity scaling by Bose symmetry studied in~\cite{McGadyRodina}.

\subsubsection{Eliminating the node dependence}
In the free polynomial ring on the $z_i$ and the coefficients of $T$, the factors $z_j-z_i$ are irreducible. Regularity at a generic point of each diagonal therefore cancels the entire denominator of~\eqref{eq:bcfw-pullback}. To determine the remaining polynomial dependence, send a single node $z_a$ to infinity at fixed $T$ and fixed spectator nodes. Equation~\eqref{eq:bcfw-map} gives
\[
 \Delta=O(z_a^{n-1}),\qquad Y_a=O(z_a^{-2}),\qquad
 Y_j=O(z_a^{-1})\quad(j\ne a).
\]
Choose a polynomial representative with the prescribed degrees and weights. A monomial with $b_a$ right spinors on leg $a$ has $b_a+n-5$ left spinors there, and $2d-b_a$ right spinors on the other legs. Its contribution to $P$ grows at most as $z_a$ to the power
\[
 (n-1)+(b_a+n-5)-2b_a-(2d-b_a)=0.
\]
A polynomial bounded as $z_a\to\infty$ is independent of $z_a$. Applying this to every leg leaves a homogeneous polynomial $P(T)$ of total degree $2d$ (or zero).

Now take $T(u)=(u-u_0)S(u)$, so $T_1$ and $T_2$ have a common finite root. Hold $T$ fixed, set $z_a=u_0$, and let $z_b\to u_0$ with all other nodes distinct. Then $Y_a=0$ and
\[
 Y_b=\frac{S(z_b)}{\prod_{k\ne a,b}(z_b-z_k)}.
\]
All original spinors have finite limits. Hence the polynomial $N_n$ is finite while $\Delta\to0$. Node independence yields $P(T)=0$. Evaluation at $Y_a=0$ is legitimate here because the identity for distinct nodes extends polynomially in $T$.

Homogenize $T_1,T_2$ to binary forms of degree $d$. Their common-root locus is the irreducible resultant hypersurface, of bidegree $(d,d)$. Pairs sharing a finite root are dense in it. Thus $\operatorname{Res}(T_1,T_2)$ divides $P(T)$, and equality of total degrees gives
\begin{equation}
 P(T)=\kappa_n\operatorname{Res}(T_1,T_2),\qquad
 \mathcal A_n^{\mathrm{aff}}
 =\kappa_n\frac{\operatorname{Res}(T_1,T_2)}{\Delta(z)^2}.
 \label{eq:bcfw-resultant}
\end{equation}
Here the scalar multiplies the resultant in a chosen sign convention. The relative-normalization formulas below use the coefficient of the Hodges determinant.
This proves an upper bound of one on the dimension. The dense-chart identity determines the original polynomial class because the momentum ring is a domain (Proposition~\ref{prop:moment-domain}).

\subsubsection{Identifying the surviving line}
We match~\eqref{eq:bcfw-resultant} to the Hodges determinant~\cite{Hodges2012}. The rational-map formula for the gravity determinant in terms of a resultant is given in~\cite[Sec.~3]{CachazoResultants}. In our conventions write the B\'ezout matrix $B$ as
\[
 \frac{T_1(u)T_2(v)-T_2(u)T_1(v)}{v-u}
 = e(u)^{\mathsf T}B e(v)=K(u,v),\qquad
 e(u)=(1,u,\ldots,u^{d-1})^{\mathsf T}.
\]
Off the diagonal, the Hodges matrix $\Phi_{ij}=[ij]/\ang{i}{j}$ equals $K(z_i,z_j)/(h'(z_i)h'(z_j))$. Choosing the same reference spinor $(0,1)^{\mathsf T}$ twice makes the row sums vanish. Since $\deg_v K\le d-1$, Lagrange interpolation gives $\sum_j K(z_i,z_j)/h'(z_j)=0$, so the formula also reproduces the diagonal entries. Keep the rows and columns in a set $I$ of size $d$, and let $V_I$ have rows $e(z_i)^{\mathsf T}$ for $i\in I$ in increasing order. Then
\[
 \det\Phi_I=\frac{\det(V_I)^2\det B}{\prod_{i\in I}h'(z_i)^2},
 \qquad \Delta(z)^2\det'\Phi_n=\det B.
\]
The second identity follows on dividing the minor by the squared Vandermonde of the three deleted nodes. The classical B\'ezout identity gives $\det B=\pm\operatorname{Res}(T_1,T_2)$, with a fixed sign for each $d$. Thus the class $N_n^{\mathrm H}$, polynomial by Lemma~\ref{lem:hodges-polynomial}, is nonzero. In~\eqref{eq:bcfw-collision}, $T_\epsilon$ stays finite, and $\mathcal A_n(Z)=\epsilon^4\mathcal A_n^{\mathrm{aff}}=O(\epsilon^2)$. The Hodges class meets the hypothesis and attains the bound in Theorem~\ref{thm:bcfw}.

\paragraph{Improved falloff on every pair.}
The stronger decay gives $P=O(1)$ on each tested diagonal in~\eqref{eq:bcfw-collision}. The spectator argument then removes all poles directly, and the degree and common-root arguments apply verbatim. Both statements leave an independent scalar at each multiplicity. After fixing a lower-point normalization, one nonzero collinear residue evaluated at a generic boundary configuration fixes the relative scalar.

\section{Worked evaluations and computational checks}
\label{app:computations}
\subsection{Worked determinant evaluations}
\label{app:worked-details}
The following evaluations determine the determinant normalizations and shift coefficients.

\subsubsection{Integer configurations and retained Hodges minors}
\label{app:worked-witness}
We use the seven- and eight-point data in~\eqref{eq:integer-cubic-witness} and the factorization~\eqref{eq:counter-rank-three}.

Momentum conservation is
\begin{equation}
 \sum_iX_iY_i^{\mathsf T}
 =\begin{pmatrix}\sum_i c_i&\sum_i c_iz_i^3\\
                  \sum_i c_iz_i&\sum_i c_iz_i^4\end{pmatrix}
 =\begin{pmatrix}m_0&m_3\\m_1&m_4\end{pmatrix}=0.
 \label{eq:worked-momentum-matrix}
\end{equation}
For example, the eight-point table gives
\begin{equation}
 \begin{aligned}
 m_0&=1-7+21-35+35-21+7-1=0,\\
 m_4&=-7+336-2835+8960\\
    &\hspace{1em}-13125+9072-2401=0.
 \end{aligned}
 \label{eq:worked-momentum-arithmetic}
\end{equation}
Identity~\eqref{eq:integer-moments} gives the other two components and $m_2=0$, which fixes the Hodges diagonal.

Delete labels $1,2,3$; the remaining nodes at eight points are $3,4,5,6,7$. Divide row $i$ and column $i$ by $c_i$. The resulting matrix is
\begin{equation}
 K_5=\begin{pmatrix}
 27&37&49&63&79\\
 37&48&61&76&93\\
 49&61&75&91&109\\
 63&76&91&108&127\\
 79&93&109&127&147
 \end{pmatrix},\qquad
 (K_5)_{ij}=z_i^2+z_iz_j+z_j^2.
 \label{eq:worked-retained-hodges}
\end{equation}
Its first four rows obey $-\mathrm{row}_1+3\mathrm{row}_2-3\mathrm{row}_3+\mathrm{row}_4=0$, and the last four obey the shifted relation, since the entries are quadratic. The leading three-by-three minor has determinant $-4$, so the rank is three. The seven-point minor $K_4$ is the upper-left four-by-four block. Nonzero factors $c_i$ preserve both ranks.

The three-column factorization gives $\rank\Phi_n\le3$. With deleted-triple denominator $(\ang{1}{2}\ang{2}{3}\ang{3}{1})^2=4$, restoring all factors gives
\begin{equation}
 H_n=\frac{D_z}{4}\left(\prod_{i=4}^{n}c_i^2\right)\det K_{n-3}=0,
 \qquad n=7,8.
 \label{eq:hodges-witness-zero}
\end{equation}

\subsubsection{Seven-point marked determinant}
\label{app:worked-seven}
\paragraph{Evaluating one marked term directly from Hodges entries.}
Take $a=2$, $b=1$, $c=3$ and $q=(0,1)^{\mathsf T}$. Then $\varepsilon_{abcB}=-1$, $B=\{4,5,6,7\}$, and $\ang{a}{q}=1$. Set $d_i=z_i-1$, so the four remaining values of $d_i$ are $2,3,4,5$. In the determinant in \eqref{eq:seven-hodges-confluent}, divide each row by $c_i$ and the four columns by $c_2,c_2,c_1,c_3$, respectively. The resulting row of a matrix $E_7$ is
\begin{equation}
 (E_7)_i=\bigl(P(d_i),P(d_i)/d_i,Q(d_i),T(d_i)\bigr),
 \quad
 \begin{aligned}
 P(d)&=d^2+3d+3,\\
 Q(d)&=d^2+2d+1,\\
 T(d)&=d^2+4d+7.
 \end{aligned}
 \label{eq:seven-worked-row}
\end{equation}
Move the derivative column to the end; two column swaps give no sign. The polynomial columns $(P,Q,T)$ have coefficient determinant $2$ relative to $(1,d,d^2)$. Subtract $d+3$ from the last column, using $P/d=d+3+3/d$, and factor out $3$. Multiplying each remaining row by $d_i$ gives
\begin{equation}
 \begin{aligned}
 \det E_7
 &=\frac{6}{2\cdot3\cdot4\cdot5}
 \det\begin{pmatrix}
 2&4&8&1\\3&9&27&1\\4&16&64&1\\5&25&125&1
 \end{pmatrix}\\
 &=\frac{6}{120}\,(-12)=-\frac35.
 \end{aligned}
 \label{eq:seven-worked-minor}
\end{equation}
Three swaps put the integer matrix in Vandermonde order on $2,3,4,5$, with product of differences $12$.

Write $V_B=\prod_{i<j,\ i,j\in B}(z_j-z_i)=12$. Separating the nodes $0,1,2$ from $B$ gives
\begin{equation}
 D_z=2V_B\prod_{i\in B}z_i d_i(z_i-2),\qquad
 \det E_7=-\frac{6V_B}{\prod_{i\in B}d_i}.
 \label{eq:worked-vandermonde-split}
\end{equation}
All row and column factors in \eqref{eq:seven-hodges-confluent} then reduce to
\begin{equation}
 \frac{N_2}{P_cD_z}
 =-c_2\frac{\prod_{i\in B}d_i^2z_i(z_i-2)}{D_z}\det E_7
 =3c_2=-18.
 \label{eq:seven-worked-restoration}
\end{equation}
The seven-point version of \eqref{eq:integer-cubic-witness} has $c=(1,-6,15,-20,15,-6,1)$. On $x=1$, write $Y=c(1,w)$ and differentiate at fixed $w$ before substituting $w=z^3$. The eight columns of the row are $1,z,z^2,z^3,w,wz,wz^2,wz^3$. Replace $w$ by $w-z^3$ and move this relation column to the end. The seven ordinary rows then give the Vandermonde on $1,z,\ldots,z^6$, while the added derivative reads $-3z_a^2$ in the last column. Moving that column through three columns cancels this minus sign. Restoring the row factors yields
\begin{equation}
 \begin{aligned}
 \frac{N_a}{P_cD_z}&=3c_az_a^2,\\
 \frac{(N_1,\ldots,N_7)}{P_cD_z}
 &=(0,-18,180,-540,720,-450,108).
 \end{aligned}
 \label{eq:seven-integer-values}
\end{equation}
Their sum vanishes by $m_2=0$. Consequently $F_7/(P_cD_z)=18\ne0$, whereas $H_7=0$. For every $\lambda\ne0$, the global polynomial $N_7^{(\lambda)}$ is therefore not proportional to $H_7$.

\subsubsection{Eight-point marked determinant}
\label{app:worked-eight}
\paragraph{The five-by-five correction at the same marked label.}
Keep $a=2,b=1,c=3$ and $q=(0,1)^{\mathsf T}$. Now $B=\{4,5,6,7,8\}$ and $d_i=2,3,4,5,6$. Divide the rows of the determinant in \eqref{eq:eight-hodges-confluent} by $c_i$, and its columns by $c_2,c_2,c_2,c_1,c_3$. With the same polynomials $P,Q,T$ as in \eqref{eq:seven-worked-row}, the resulting matrix has rows
\begin{equation}
 (E_8)_i=\bigl(P(d_i),P(d_i)/d_i,2P(d_i)/d_i^2,Q(d_i),T(d_i)\bigr).
 \label{eq:eight-worked-row}
\end{equation}
Reorder the columns as $(P,Q,T,P/d,2P/d^2)$; four swaps give a positive sign. The first three again contribute the factor $2$. In the last two columns,
\[
 \frac{P}{d}=d+3+\frac3d,\qquad
 \frac{2P}{d^2}=2+\frac6d+\frac6{d^2}.
\]
Subtract the polynomial terms, then subtract twice the reduced fourth column from the fifth. The remaining columns are $3/d$ and $6/d^2$. Factoring out $2\cdot3\cdot6=36$ and multiplying every row by $d_i^2$ gives
\begin{equation}
 \begin{aligned}
 \det E_8
 &=\frac{36}{(2\cdot3\cdot4\cdot5\cdot6)^2}
 \det\begin{pmatrix}
 4&8&16&2&1\\9&27&81&3&1\\16&64&256&4&1\\
 25&125&625&5&1\\36&216&1296&6&1
 \end{pmatrix}\\
 &=\frac{36}{720^2}\,(-288)=-\frac1{50}.
 \end{aligned}
 \label{eq:eight-worked-minor}
\end{equation}
The powers $(2,3,4,1,0)$ have odd parity relative to $(0,1,2,3,4)$, and the Vandermonde product is $1!2!3!4!=288$.

The factorization of $D_z$ in \eqref{eq:worked-vandermonde-split} remains valid with this five-element $B$, now with $V_B=288$. Hence
\begin{equation}
 \begin{aligned}
 \det E_8&=-\frac{36V_B}{\prod_{i\in B}d_i^2},\\
 \frac{R_2}{P_cD_z}
 &=-c_2^2\frac{\prod_{i\in B}d_i^3z_i(z_i-2)}{D_z}\det E_8
 =18c_2^2=882.
 \end{aligned}
 \label{eq:eight-worked-restoration}
\end{equation}
The symmetric correction is the sum over all marks.

For this sum, use the eight-point data
\[
 c=(1,-7,21,-35,35,-21,7,-1),\qquad z=(0,1,2,3,4,5,6,7).
\]
The ten columns, after factoring out $c_i$, are $1,z,z^2,z^3,z^4,w,wz,wz^2,wz^3,wz^4$. A determinant-one column operation puts them in the order
\[
 1,z,z^2,z^3,z^4,wz^2,wz^3,wz^4,\quad w-z^3,\quad z(w-z^3).
\]
On $w=z^3$, the first eight columns give $1,z,\ldots,z^7$ and the last two vanish on every original row. Write $f=w-z^3$ and $g=z(w-z^3)$. Differentiating before restricting to the cubic curve gives
\begin{equation}
 \begin{array}{c|cc}
     &f&g\\\hline
 \partial_z&-3z^2&w-4z^3\\
 \partial_z^2&-6z&-12z^2
 \end{array}
 \label{eq:eight-worked-derivatives}
\end{equation}
On $w=z_a^3$, the residual two-by-two determinant is

\begin{equation}
 \det\begin{pmatrix}-3z_a^2&-3z_a^3\\-6z_a&-12z_a^2\end{pmatrix}
 =18z_a^4=2(3z_a^2)^2.
 \label{eq:eight-residual-two-by-two}
\end{equation}
Restoring the original row factors $P_cD_z$ and the derivative factors $c_a^2$ gives
\begin{equation}
 \frac{R_a}{P_cD_z}=18(c_az_a^2)^2.
 \label{eq:eight-worked-allmarks}
\end{equation}
At $a=2$, the block $\left(\begin{smallmatrix}-3&-3\\-6&-12\end{smallmatrix}\right)$ has determinant $18$; $c_2^2=49$ gives $882$, as in \eqref{eq:eight-worked-restoration}. Summing all marks gives
\begin{center}
\small
\begin{tabular}{c@{\qquad}r@{\qquad}r}
\toprule
$a$ & $c_az_a^2$ & $(c_az_a^2)^2$\\
\midrule
1&0&0\\2&$-7$&49\\3&84&7056\\4&$-315$&99225\\
5&560&313600\\6&$-525$&275625\\7&252&63504\\8&$-49$&2401\\
\midrule
\multicolumn{2}{r}{sum}&761460\\
\bottomrule
\end{tabular}
\end{center}
The total is $18\cdot761460=13706280$, reproducing~\eqref{eq:eight-counterexample}.
All $c_a$ and all node differences are nonzero, so $P_cD_z\ne0$. Thus every $\lambda\ne0$ in \eqref{eq:eight-hodges-family} gives a numerator with the same degrees, all pair zeros and Bose symmetry, but not a constant multiple of $H_8$.

\subsubsection{Interpolation of the eight-point shift}
\label{app:worked-shift}
Use the shift and data specified in Section~\ref{subsec:shift}.
\paragraph{The denominator and its leading coefficient.}
Only the six factors $\ang{2}{j}$ with $j=3,\ldots,8$ depend on $Z$, and
\[
 \ang{2}{j}(Z)=(j-2)-(j-1)Z.
\]
Removing their values at $Z=0$ leaves the constant
\begin{equation}
 \begin{gathered}
 C=\frac{D_8(0)}{6!}=174182400,\qquad [Z^6]D_8=5040C,\\
 D_8(Z)=C\prod_{j=2}^{7}\bigl((j-1)-jZ\bigr).
 \end{gathered}
 \label{eq:worked-shift-denominator}
\end{equation}
\paragraph{Extracting the numerator coefficients.}
The leg degrees bound each shifted $\mathcal R$ and $\mathcal S$ by degree seven in $Z$. For example, in $R_1$ the shifted right-spinor degree is three and the shifted angle-spinor degree is four; in $R_2$ they are one and six. The two-mark terms in $\mathcal S$ obey the same upper bound. Thus eight determinant evaluations suffice. For $F=\mathcal R,\mathcal S$, interpolate as
\begin{equation}
 [Z^k]F(Z)=\sum_{\ell=2}^{9}F(\ell)[Z^k]
 \prod_{\substack{m=2\\m\ne\ell}}^{9}\frac{Z-m}{\ell-m}.
 \label{eq:worked-shift-interpolation}
\end{equation}
The nodes avoid the chart boundary $x_2=0$ at $Z=1$. Table~\ref{tab:worked-shift} gives the resulting coefficients. Since $[Z^7]F=0$, the seventh forward difference vanishes and $[Z^6]F=\Delta^6F(2)/6!$.
\begin{table}[htbp]
\centering\small\renewcommand{\arraystretch}{1.35}
\begin{tabular}{cccc}
\toprule
$k$ & $[Z^k]\mathcal R/C$ & $[Z^k]\mathcal S/C$ & $[Z^k]D_8/C$\\
\midrule
7&$0$&$0$&$0$\\
6&$-317$&$1902$&$5040$\\
5&$\frac{1161103}{840}$&$-\frac{1161103}{140}$&$-22212$\\
0&$-\frac{5129}{63}$&$\frac{35182}{75}$&$720$\\
\bottomrule
\end{tabular}
\caption{Exact coefficients on the specified BCFW path, after removing the common factor $C$.}
\label{tab:worked-shift}
\end{table}

\subsection{Nine-point certificate and arithmetic checks}
\label{app:checks}
At nine points, a finite rank certificate checks the large-$Z$ condition in characteristic zero. Each pair-zero class has a representative containing at least one $Y_i$ from every leg (Theorem~\ref{thm:soft-lift}). Alternating representatives of this kind have local spinor degrees $(5+e_i,1+e_i)$ with $\sum_i e_i=3$. The distributions $3$, $2+1$, $1+1+1$ contribute $1,34,87$ invariant directions before the momentum quotient. The exact momentum Koszul complex has invariant dimensions $122,70,11,1$, leaving a $62$-dimensional quotient that can contain candidates failing the pair conditions.

For a generic shift $D_9(Z)$ has degree seven, so decay requires the seventh forward difference of $N_9(Z)$ to vanish. Using $72$ momentum-conserving samples and $122$ free invariant columns gives rank $61$ modulo $1009$; a $61\times61$ minor has determinant $653$ in that field. The invariant basis reduces surjectively because the raising-operator matrices have the same rank in characteristic zero and modulo $1009$. The sample points lift on a rational chart with denominators prime to $1009$. The minor therefore bounds the characteristic-zero rank from below, leaving at most one dimension, attained by the Hodges solution. A separate pair-zero restriction matrix has rank $57$ and gives $1\le\dim W_9^{\sgn}\le5$ before the shift.

Exact rational checks for $5\le n\le12$ match the reduced Hodges minor, B\'ezout determinant and signed resultant, both before and after the same ordered BCFW shift at $Z=1/3$ and $Z=-2$ at each multiplicity (\texttt{check\_resultant.py}).

\subsubsection{Exact evidence and the earlier formal checks}
\label{subsec:verification-scope}
The arithmetic scripts check determinant conventions and boundary coefficients using exact rational or finite-field operations. Modular rank bounds require an integral model and denominator control, as above; see~\cite{Peraro2016} for finite-field reconstruction in amplitudes. The determinant formulas supply polynomial representatives. Completeness at seven and eight points, and Theorem~\ref{thm:bcfw} on the full $V_n^{\sgn}$, follow from the analytic arguments independently of these checks.

The first arXiv version used flag-variety standard monomials and permutation-sector restriction matrices~\cite{MaiZhangV1}. Its Lean~4/mathlib companion~\cite{DeMouraUllrich2021,MathlibCommunity2020} assumed the independently established seven-point decomposition and eight-point plane, in a different coordinate basis. It verified removal of the hook by Bose symmetry, kernels of the physical comparison maps, and uniqueness conditional on boundary injectivity. These checks do not construct the candidate spaces or prove the physical theorems, detection, lifting or resultant arguments presented here.

For related formalizations and discussion of their premises, see~\cite{ToobySmith2025HepLean,ToobySmith2025Wick,DouglasEtAl2026FormalQFT,Douglas2026PhysicalReasoning}. The first version and project repository link the earlier implementation. The present ancillary files contain the arithmetic scripts, compact nine-point matrix, selected minor, generation code and field-lifting details.

\end{document}